\documentclass{siamltex}

\usepackage{amsmath}
\usepackage{amssymb}
\usepackage{ifthen}

\usepackage[left=2cm,right=2cm,top=2cm,bottom=2cm]{geometry}
\usepackage{color}
\usepackage{graphicx}
\usepackage{subcaption}

\usepackage{changebar}
\usepackage{booktabs}

\newcommand{\reals}{ \ensuremath{ \mathbb{R} }}

\newcommand{\leqr}{ \ensuremath{ := } }

\numberwithin{equation}{section}

\newboolean{draft}
\setboolean{draft}{false} 

\ifthenelse{\boolean{draft}}{\usepackage{showkeys}}{}

\title{On the Hierarchy of 
Scales in
Modeling of
Weakly Interacting Chains
of Atoms
}

\author{
    Dmitry Golovaty\footnotemark[1]
\and
    J. Patrick Wilber\footnotemark[1]}

\begin{document}
\maketitle

\renewcommand{\thefootnote}{\fnsymbol{footnote}}
\footnotetext[1]{
    Department of Mathematics,
    University of Akron,
    Akron,
    OH 44325, USA.}

\begin{abstract}
In the first part of this paper, we apply a well known discrete-to-continuum approach to a Frenkel-Kontorova-type model of an infinitely long one-dimensional chain of atoms weakly interacting with a line of fixed atoms.  The rescaled model contains a small parameter $\delta$ that is the ratio of the strengths of the weak interaction and the elastic interaction.  After replacing discrete displacements with piecewise affine functions to define continuum versions of the discrete energies, we prove that these energies $\Gamma$-converge to a continuum energy as $\delta\rightarrow 0$.  This limiting process represents a transition from the microscale, at which individual atoms are resolved, to a mesoscale with a single diffuse domain wall.  In the second part of this paper, we introduce an additional rescaling $\varepsilon$, and an associated limiting process that converts our problem to the macroscale.  The $\varepsilon$-limiting energy is finite for piecewise constant functions of bounded variation.  In the context of our problem, each point of discontinuity of a minimizer of the limiting energy corresponds to a sharp domain wall.
\end{abstract}


\section{Introduction} 

In recent years, there has been an extensive effort to model registry
effects in suspended graphene, bilayer graphene, and related layered
nanostructures
\cite{carr2020electronic,carr2019exact,carr2018relaxation,cazeaux2020energy,
hott2024incommensurate,massatt2023electronic,zhang2018structural}.
As part of this effort, researchers have sought
basic insight into phenomena like the formation of domain walls,
localized out-of-plane displacements,
and relaxed moire patterns using one-dimensional discrete models
\cite{cazeaux2017analysis,espanol2017discrete,golovaty2008continuum,jingzhi2024formaljustificationcontinuumrelaxation,PhysRevB.96.075311,popov2011commensurate}.
Some of this work applies the framework of the classical
Frenkel-Kontorova theory \cite{braun2013frenkel}, which, for slightly
mismatched one-dimensional lattices, predicts the formation of
relatively large commensurate regions separated by localized
incommensurate regions.
The simplified setting of one-dimensional models facilitates the
application of discrete-to-continuum, or upscaling, 
procedures whose resulting continuum models can yield further insight.

We formulate a model of an infinitely long one-dimensional chain of
atoms weakly interacting with a line of fixed atoms.  Nearest
neighbors on the chain interact by linear springs.  The weak
interaction is minimized at positions on the chain above the midpoints
of the fixed atoms.  This set up gives a Frenkel-Kontorova-type model,
in which the positions of atoms on the chain are determined by a
competition between elastic interactions with nearest neighbors and
the potential energy wells from the weak interaction.  Displacement
boundary conditions are imposed at $\pm \infty$ that preclude the
system from attaining global registry.

In the first part of this paper, we apply a well known
discrete-to-continuum approach (see, for example,
\cite{braides2002gamma,braides2006discrete,scardia2016continuum}) to study our model.  By rescaling, we
introduce a small parameter $\delta$ that is the ratio of the
strengths of the weak interaction and the elastic interaction.
Following \cite{braides2002gamma}, we replace discrete displacements
with piecewise affine functions to define continuum versions of the
discrete energies.  We prove that these energies $\Gamma$-converge to
a continuum energy as $\delta\rightarrow 0$.

The first variation of the limiting continuum energy yields a
boundary-value problem for the diplacement of points on the continuum
chain.  This boundary-value problem implies that a typical minimizer
of the limiting energy is monotone and approaches constant integer
values at $\pm \infty$, where these integers differ by 1.  Hence the
minimizer has a single domain wall between unbounded regions of
registry on the left and on the right.  The domain wall is
spatially diffuse.  Based upon this, we interpret $\delta$ as a length
scale associated with a single domain wall.  At this length scale the
transition appears smooth rather than sharp and the individual
domain walls are infinitely far apart, which is why we see only a
single transition.  The limiting process takes any additional domain
walls of the discrete energies and pushes these to infinity.
We can interpret the limiting process based on 
$\delta$ as rescaling the problem from a microscale, at which one sees
individual atoms, to a mesoscopic scale, at which the atoms have been
homogenized and the chain appears as a continuous curve.  However,
this scale is still relatively small because we see only a single
domain wall.

In the second part of this paper, we introduce an additional rescaling
and an associated limiting process that converts our problem to the
macroscale.
To motivate this, at the start of Section~\ref{s5}
we observe that our new rescaling, 
when applied to the
limiting energy from the $\Gamma$-convergence result of the first part of
the paper,
introduces another small parameter, $\varepsilon$,
and generates a new family of energies having the structure of the
Modica-Mortola energy \cite{modica1987gradient,leoni2013gamma}.
%
The $\varepsilon$-limiting energy is finite for piecewise constant
functions of bounded variation.  In the context of our problem, each point of discontinuity
corresponds to a sharp domain wall.
If we consider more general displacement boundary conditions, for
which the limiting values of the displacement at $\pm \infty$ differ
by a prescribed integer, then minimizers can exhibit multiple points
of discontinuity and hence multiple domain walls.
We interpret the limiting process of letting $\varepsilon\rightarrow
0$ as going from a mesoscale view of the system to macroscale view.
Instead of seeing a single spatially diffuse domain wall, we see
multiple domain walls separated by a finite distance.  At the
macroscale, we no longer resolve the details of individual walls.
Hence these no longer appear as smooth transitions but instead now
appear as jump discontinuities.

Motivated by these observations, we introduce the additional length
scale $\varepsilon$ into the original discrete energies for the chain,
and we consider more general displacement boundary conditions that
allow the change in the displacement from $\pm\infty$ 
to be any integer.
This yields a 2-parameter collection of discrete energies each
depending on $\delta$ and $\varepsilon$.  Each of these energies has
the structure of the Modica-Mortola functional.  We then apply
the discrete-to-continuum procedure from the first part of the paper to
this collection to find the $\Gamma$-limit of these energies.

The proof of $\Gamma$-convergence in the second part of this paper is
similar in structure and in some details to the proof of the
$\Gamma$-convergence result in \cite{leoni2013gamma}.  However, our
proof incorporates the discrete-to-continuum procedure in
\cite{leoni2013gamma}.  In addition, the problem we study here differs
from that in \cite{leoni2013gamma} because it is on an unbounded
spatial domain and our potential is periodic with an infinite number
of energy wells.  These differences, in particular, necessitate a more
complicated construction for the recovery sequence.

In the next section, we derive our discrete energies and rescale to
introduce a small parameter.  Section~\ref{s3}, we prove a
$\Gamma$-convergence result for these energies.  At the start of
Section~\ref{s5}, we motivate the introduction of another small
parameter into the discrete energies.  We then prove
$\Gamma$-convergence result for the 2-parameter family of discrete
energies.  The final section contains some concluding comments.

\section{Description of Discrete Problem} \label{s1}

We model a one-dimensional chain of atoms parallel to and a constant
distance $\bar{s}$ from an infinite line of fixed atoms.  See
Figure~\ref{ff1}.  The fixed atoms are a distance $l$ apart.  Atoms on
the chain can displace horizontally.  Neighboring atoms interact
elastically, and every atom on the chain weakly interacts with every
atom on the lower line.  This weak interaction models van der Waals
forces between the atoms.  We assume that in the reference
configuration the position of atom $i$ on the chain is $il+l/2$ for
$i\in \mathbb{Z}$.  In the deformed configuration of the chain, atom
$i$ has displacement $u_{i}$.

\begin{figure}[htb]
\centering
  \includegraphics[scale = 0.5, clip, trim=0in 1.5in 0in 2in]{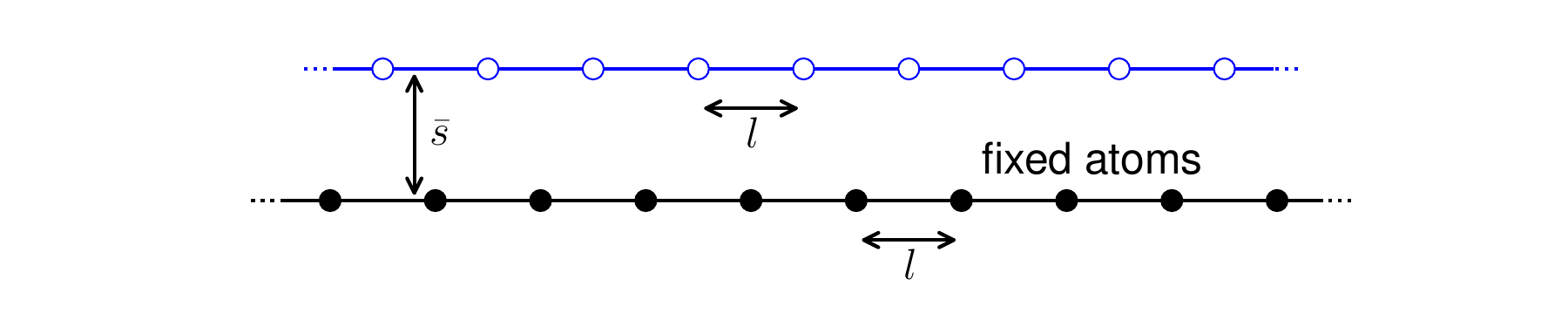}
\caption{Reference Configuration of Discrete System.}
  \label{ff1}
\end{figure}

We associate with the displacements $\{ u_{i}  \}$ the discrete energy
\begin{align}
  \hat{E}[\{ u_{i}\}]
  &=
  a \sum_{i=-\infty}^{\infty}\left(\frac{u_{i}-u_{i-1}}{l}\right)^{2}
  \nonumber \\
  &\phantom{=}+
  b \sum_{i=-\infty}^{\infty} \sum_{j=-\infty}^{\infty}
  \left[
    h\left( \frac{il+\frac{1}{2}l+u_{i}-jl}{p} \right)
    -
    h\left( \frac{il+\frac{1}{2}l-jl}{p} \right)
  \right], 
  \label{ee1}
\end{align}
where $a$ and $b$ are positive constants.
The first sum is the total elastic energy between
neighboring atoms on the chain.
For the second sum, $h$ describes the pairwise weak interaction energy
as a function of horizontal position on the chain relative to the
$j$th fixed atom.  (See the discussion following \eqref{ee8} below.)
The constant $p$ is a length scale associated with $h$.

To avoid boundary effects, we consider a system of infinite size.
In this section and the next, we assume that 
\begin{equation}
  \underset{i\rightarrow -\infty}{\lim}u_{i}=0,
  \qquad
  \underset{i\rightarrow \infty}{\lim}u_{i}=l.
  \label{ee19}
\end{equation}
These boundary conditions prevent the system from attaining global registry.

Next, we rescale the distance between atoms by $l$
and we rescale the other variables as
\begin{equation*}
  \xi_{i} = \frac{u_{i}}{l},
  \qquad
  \sigma=\frac{p}{l},
  \qquad
  E=\frac{1}{\sqrt{ab}}\hat{E}.
\end{equation*}
Using these in \eqref{ee1} yields
\begin{align*}
  E[\{ \xi_{i}  \}]
  &=
  \sqrt{a/b}\sum_{i=-\infty}^{\infty}\left(\xi_{i}-\xi_{i-1}\right)^{2}
  \nonumber \\
  &\phantom{=}+
   \sqrt{b/a}\sum_{i=-\infty}^{\infty} \sum_{j=-\infty}^{\infty}  
  \left[
    h\left( \frac{i+\frac{1}{2}+\xi_{i}-j}{\sigma} \right)
    -
    h\left( \frac{i+\frac{1}{2}-j}{\sigma} \right)
  \right] \nonumber  \\
  &=
  \sqrt{b/a}\sum_{i=-\infty}^{\infty}
  \left(\frac{\xi_{i}-\xi_{i-1}}{\sqrt{b/a}}\right)^{2}
  \nonumber \\
  &\phantom{=}+
   \sqrt{b/a}\sum_{i=-\infty}^{\infty} \sum_{j=-\infty}^{\infty}  
  \left[
    h\left( \frac{i+\frac{1}{2}+\xi_{i}-j}{\sigma} \right)
    -
    h\left( \frac{i+\frac{1}{2}-j}{\sigma} \right)
  \right]. 
\end{align*}
We define $\delta=\sqrt{b/a}$ and 
\begin{align}
  E^{\delta}[\{ \xi_{i}  \}]
  &=
  \delta\sum_{i=-\infty}^{\infty}
   \left(\frac{\xi_{i}-\xi_{i-1}}{\delta}\right)^{2}
  \nonumber \\
  &\phantom{=}+
   \delta\sum_{i=-\infty}^{\infty} \sum_{j=-\infty}^{\infty}  
  \left[
    h\left( \frac{i+\frac{1}{2}+\xi_{i}-j}{\sigma} \right)
    -
    h\left( \frac{i+\frac{1}{2}-j}{\sigma} \right)
  \right].
  \label{ee6}
\end{align}
%


Note that, because the chain is infinite, there is no natural scale
associated with its length.  On the other hand, the
problem has a dimensionless parameter, $\delta$, the ratio of the
strength of elastic interaction to the strength of the weak
interaction.
In the Riemann sums \eqref{ee6}, $\delta$ plays the role of a
discretization length.
This suggests that variations of discrete energy minimizers
$\{\xi_i\}$ should be of order $1$ on intervals of the length
$\sim1/\delta$.
We can interpret this fact by assuming that we are
viewing the chain with a unit spacing between the atoms
from a distance $\sim1/\delta$, and hence as
$\delta\to 0$, $\xi$ would appear as a continuous functions with $O(1)$
variations over intervals of length $1$.
As we will discuss below, from this point of view, the limiting process of $\delta$ going to
zero corresponds to moving away from the chain at a specific rate.  If
we move away from the chain at a rate faster than $1/\delta$, we would
observe sharp transitions in the
function $\xi$.
%


\section{Convergence to a Mesoscale Model} \label{s3}

In this section, we establish a $\Gamma$-convergence result for a
continuum version of $E^{\delta}$ as $\delta\rightarrow 0$.  The
discrete-to-continuum approach we use here follows
\cite{braides2006discrete, scardia2016continuum}.

To adapt the boundary conditions \eqref{ee19} to the continuum setting, 
we let $\bar{v}$ be a continuously differentiable, strictly increasing
function satisfying 
\begin{gather*}
  \bar{v}(x)
  \rightarrow 
  \begin{Bmatrix}
   1\\
   0
  \end{Bmatrix}
  \ \text{as}\ 
  x\rightarrow 
  \begin{Bmatrix}
   \infty\\
   -\infty
  \end{Bmatrix}, \\[2mm]
  1-\bar{v}\in L^{1}([0,\infty))\cap L^{2}([0,\infty)),
  \
  \bar{v}\in L^{1}((-\infty,0])\cap L^{2}((-\infty,0]),
  \
  \bar{v}'\in L^{1}(\reals)\cap L^{2}(\reals).
\end{gather*}
(For example, we could use $\bar{v}(x)=(\tanh x+1)/2$.)
We define the ambient space $A$ by
\begin{equation}
  A
  =
  \{ \xi\, :\, \xi-\bar{v}\in
  H^{1}(\reals) \},
  \label{ee21}
\end{equation}
and we say that $\{ \xi_{n} \}$ converges weakly to $\xi$ in
$A$ if $\xi_{n}-\bar{v}\rightharpoonup\xi-\bar{v}$ in
$H^{1}(\reals)$ as $n\rightarrow \infty$.
Note that this means in particular that
$\xi_{n}-\bar{v}$ converges to
$\xi-\bar{v}$ strongly in $L^{2}(\reals)$, and hence
$\|\xi_{n}-\xi\|_{2}\rightarrow 0$ as $n\rightarrow 0$.
We define a continuum energy $F:A\rightarrow  \reals$ by
\begin{equation}
  F[\xi]
  =
  \int_{-\infty}^{\infty}\! \xi'(x)^{2} \,dx
  +
  \int_{-\infty}^{\infty}\!w(\xi(x)) \,dx,
  \label{ee7}
\end{equation}
where 
\begin{equation}
  w(\xi)
  =
  \sum_{j=-\infty}^{\infty}\left[
  h\left( \frac{j+\frac{1}{2}+\xi}{\sigma} \right)
  -
  h\left( \frac{j+\frac{1}{2}}{\sigma}\right) \right].
  \label{ee8}
\end{equation}
The function $w$ has period 1 and $w(z)=0$ for all
$z\in\mathbb{Z}$.  We assume that $h$ is an even function, from which it
can be shown that $w$ is even about $1/2$ and that $w'(z)=0$ for all
$z\in \mathbb{Z}$.  We further assume that $h$ is such that $w(\xi)>0$ for
all $\xi\notin\mathbb{Z}$, that $w''(z)>0$ for all $z\in \mathbb{Z}$,
and that $w'''$ is continuous on $\reals$.
(These would be true, for example, if
$h(d)=V(\sqrt{d^{2}+\bar{s}^{2}})$, where $V$ is a
Lennard-Jones $12$-$6$ potential and $\sigma$ and $\bar{s}$ are chosen
appropriately.)
Note that $w$ attains its minimum value of $0$ at each
$z\in\mathbb{Z}$.

For $\delta>0$, we define the partition $P_{\delta}$ of
$\reals$ by
\begin{equation}
  P_{\delta}=\{ \ldots,-2\delta,-\delta,0,\delta,2\delta,\ldots  \},
  \label{ee20}
\end{equation}
we define the function space $A_{\delta}$ by
\begin{equation*}
  A_{\delta}
  =
  \{ \xi\, :\, \xi\in A,\ \xi\ \text{is continuous},
  \xi\ \text{is piecewise affine on}\ P_{\delta}\},
\end{equation*}
and we
define the functional $E_\delta:A\rightarrow  \reals$ by
\begin{equation}
  E_\delta[\xi]
  :=
  \left\{
  \begin{array}{ll}
    \int_{-\infty}^{\infty}\xi^\prime(x)^{2}\,dx
    +\int_{-\infty}^{\infty}w\left(\xi(x)\right)\,dx, & \xi\in A_\delta, \\
    \infty, & \xi\in A\setminus A_\delta.
  \end{array}
  \right. \label{ee11}
\end{equation}

Below we prove that the energies $E_{\delta}$ $\Gamma$-convergence to
the continuum energy $F$ as $\delta \rightarrow 0$.  Our motivation
for this result is to show that the discrete energies $E^{\delta}$
defined in \eqref{ee6} converge to $F$.  However, when defining the
functional $E_{\delta}$, we replace the discrete weak interaction
energy in \eqref{ee6} by the integral expression for the weak
interaction energy in \eqref{ee11}.
Our definition of $E_{\delta}$ is consistent with the approach in
\cite{braides2006discrete, scardia2016continuum}, and in this way
$E_{\delta}$ can be defined on the space of functions $A_{\delta}$
contained in the larger ambient space $A$.
To justify interpreting our $\Gamma$-convergence result as a statement
about the discrete energies $E^{\delta}$, we present next
Lemma~\ref{l:0}, which shows that under appropriate conditions
the difference between $E_{\delta}$ and $E^{\delta}$ goes to $0$ as
$\delta \rightarrow 0$.
For this Lemma, note that with a sequence $\{\xi_i\}_{i=-\infty}^\infty\subset\mathbb{R}$ we
can associate a piecewise affine function
$\xi^\delta:\mathbb{R}\to\mathbb{R}$ on $P_{\delta}$ by setting
$\xi^\delta(i\delta)=\xi_i$ and assuming that $\xi^\delta$ is affine
on the interval $[(i-1)\delta,i\delta]$ for every $i\in\mathbb{Z}$.

\begin{lemma}
    \label{l:0}
    Let $C$ and $\delta$ be positive constants.  Let $\{\xi_{i}\}$
    be a sequence and let $\xi^\delta$ be the
    associated piecewise affine function on $P_{\delta}$.  Suppose that either
    $E^{\delta}[\{\xi_{i}\}]<C$ or that
    $E_{\delta}\left[\xi^\delta\right]<C$.  Then there is a constant
    $\hat{C}$ that depends only on $C$ such that 
    \[
    \left|
    E_{\delta}\left[\xi^\delta\right]-E^{\delta}[\{\xi_{i}\}]
    \right|
    <
    \hat{C}\delta^2.
    \]
\end{lemma}

\begin{proof}
    We can write
    \begin{equation}
    \label{eq:0.1}
    E_{\delta}\left[\xi^\delta\right]-E^{\delta}\left[\{\xi_{i}\}\right]
    =
    \sum_{i=-\infty}^{\infty}\int_{(i-1)\delta}^{i\delta}\left\{w\left(\xi_i\left(\frac{\eta}{\delta}-i+1\right)+\xi_{i-1}\left(i-\frac{\eta}{\delta}\right)\right)-\frac{1}{2}\left(w\left(\xi_{i-1}\right)+w\left(\xi_{i}\right)\right)\right\}\,d\eta\,.
    \end{equation}
    Next we estimate the integral under the summation sign. Fixing an
    $i\in\mathbb{Z}$ and setting $t=\frac\eta\delta-i+1$, we can write
    this integral as
    \begin{align}
    G^\delta_i&:=\delta\int_0^1\left\{w\left(\xi_it+\xi_{i-1}(1-t)\right)-\frac{1}{2}\left(w\left(\xi_{i-1}\right)+w\left(\xi_{i}\right)\right)\right\}\,dt\nonumber\\
    &=
    \frac{\delta}{2}\int_0^1\left\{w\left(\xi_it+\xi_{i-1}(1-t)\right)-w\left(\xi_{i-1}\right)\right\}\,dt
    +
    \frac{\delta}{2}\int_0^1\left\{w\left(\xi_it+\xi_{i-1}(1-t)\right)-w\left(\xi_i\right)\right\}\,dt\label{eq:G}\\
    &=
    \frac{\delta}{2}\int_0^1\left[\left\{w\left(\xi_{i-1}+\left(\xi_i-\xi_{i-1}\right)t\right)-w\left(\xi_{i-1}\right)\right\}
    +
    \left\{w\left(\xi_i-\left(\xi_i-\xi_{i-1}\right)t\right)-w\left(\xi_i\right)\right\}\right]\,dt\,\nonumber,
    \end{align}
    where we used the change of variables $p=1-t$ in the second
    integral on the second line and then relabeled $p$ as
    $t$. Expanding the integrand, we obtain 
    \[w\left(\xi_{i-1}+\left(\xi_i-\xi_{i-1}\right)t\right)-w\left(\xi_{i-1}\right)=w^\prime\left(\xi_{i-1}\right)\left(\xi_i-\xi_{i-1}\right)t+\frac12w^{\prime\prime}\left(\overline\xi_{i}\right){\left(\xi_i-\xi_{i-1}\right)}^2t^2\,,\]
    and
    \[w\left(\xi_{i}-\left(\xi_i-\xi_{i-1}\right)t\right)-w\left(\xi_{i}\right)=-w^\prime\left(\xi_{i}\right)\left(\xi_i-\xi_{i-1}\right)t+\frac12w^{\prime\prime}\left(\tilde\xi_{i}\right){\left(\xi_i-\xi_{i-1}\right)}^2t^2\,,\]
    so that
    \begin{align*}
      \{w(\xi_{i-1}&+(\xi_i-\xi_{i-1})t)-w(\xi_{i-1})\}
      +
      \left\{w\left(\xi_i-\left(\xi_i-\xi_{i-1}\right)t\right)-w\left(\xi_i\right)\right\}\\
      &=-\left(w^\prime\left(\xi_{i}\right)-w^\prime\left(\xi_{i-1}\right)\right)\left(\xi_i-\xi_{i-1}\right)t
      +
      \frac12\left(w^{\prime\prime}\left(\overline\xi_{i}\right)
      +
      w^{\prime\prime}\left(\tilde\xi_{i}\right)\right){\left(\xi_i-\xi_{i-1}\right)}^2t^2\\ 
      &=-w^{\prime\prime}\left(\hat\xi_{i}\right){\left(\xi_i-\xi_{i-1}\right)}^2t
      +
      \frac12\left(w^{\prime\prime}\left(\overline\xi_{i}\right)
      +
      w^{\prime\prime}\left(\tilde\xi_{i}\right)\right){\left(\xi_i-\xi_{i-1}\right)}^2t^2\,,
    \end{align*}
    where $\tilde\xi_{i},\overline\xi_{i},\hat\xi_{i}$ lie between $\xi_{i-1}$ and $\xi_{i}.$ Substituting the result into \eqref{eq:G} and integrating in $t,$ we obtain
    \[G_i^\delta=\frac\delta{12}{\left(\xi_i-\xi_{i-1}\right)}^2\left[w^{\prime\prime}\left(\overline\xi_{i}\right)+w^{\prime\prime}\left(\tilde\xi_{i}\right)-3w^{\prime\prime}\left(\hat\xi_{i}\right)\right]\,,\]
    hence
    \[\left|G_i^\delta\right|\leq
    M\delta{\left(\xi_i-\xi_{i-1}\right)}^2,\]
where $M$ is a constant that does not depend on $\{\xi_{i}\}$ or $\delta$.
Substituting this inequality into \eqref{eq:0.1} gives
    \[
    \left|
    E_{\delta}\left[\xi^\delta\right]-E^{\delta}[\{\xi_{i}\}]
    \right|
    =
    \left|
    \sum_{i=-\infty}^{\infty} G_i^\delta
    \right|
    \leq
    M\delta\sum_{i=-\infty}^{\infty}{\left(\xi_i-\xi_{i-1}\right)}^2
    \leq
    M\delta^2
    \min\left\{E_{\delta}\left[\xi^\delta\right],E^\delta\left[\{\xi_{i}\}\right]\right\}
    \leq
    \hat{C}\delta^2.\]
\end{proof}

This lemma tells us that the discrete energies $E^{\delta}[\{
  \xi_{i}^{\delta} \}]$ of a collection of sequences $\{\{
\xi_{i}^{\delta} \}\}_{\delta}$ are asymptotically the same as the energies
$E_\delta[\xi^{\delta}]$ of the collection $\{\xi^{\delta} \}_{\delta}$ of
corresponding piecewise affine functions in the limit as
$\delta\rightarrow 0$.
For the rest of this section, we work with
the energies $E_\delta[\xi^{\delta}]$ defined by \eqref{ee11}.

Next we prove a compactness result for the energies $E_{\delta}$.
As a first step, we define
\begin{equation}
  p(z)
  =
  2\int_{0}^{z}\! \sqrt{w(t)}\,dt.
  \label{ee194}
\end{equation}
Note that if $\xi\in A$ and $E_{\delta}[\xi]\leq C$, then
\begin{align}
  C
  \geq
  E_\delta[\xi]
  \geq
  2\int_{-\infty}^{\infty}\!\left|\xi'(x)\sqrt{w(\xi(x))}\right| \,dx
  =
  \int_{-\infty}^{\infty}\!\left|\frac{d}{dx}2\int_{0}^{\xi(x)}\!\sqrt{w(t)}\,dt\right| \,dx
  =
  \operatorname{Var}(p(\xi(x)),\reals).
  \label{ee196}
\end{align}
With this definition, we establish the following lemma for the
compactness result we present below.  As we see later, this lemma
allows us to handle the translational invariance in our problem.

\begin{lemma}
Let $0<\alpha<<1$ and set $I_{\alpha}=p(1-\alpha)-p(\alpha)=2\int_{\alpha}^{1-\alpha}\!\sqrt{w(t)} \,dt.$
Let $C$ satisfy $p(1)<C<\frac{3}{2}I_{\alpha}$.  Let $\xi\in A_{\delta}$
such that $E_\delta[\xi] \leq C$ with $\delta\leq 1$.  There exists
$\bar{a}, \bar{C}\in \reals$
with $\bar{C}>0$ such that,
if $\xi_{\bar{a}}(x)=\xi(x+\bar{a})$, then
\begin{equation}
  \|\xi_{\bar{a}}-\bar{v}\|_{2}
  \leq
  \bar{C},
  \label{ee195}
\end{equation}
where $\bar{a}$ is a multiple of $\delta$ and 
$\bar{C}$ depends only on $C$.    \label{r8}
\end{lemma}

\noindent \textit{Proof}. 
For $a\in \reals$, we define $\xi_{a}(x)=\xi(x+a)$ and
introduce the sets
\begin{equation*}
  V_{a}^{0}=\{ |\xi_{a}|<\alpha \}\cap \{ |1-\bar{v}|<\alpha  \},
  \quad
  V_{a}^{1}=\{ |1-\xi_{a}|<\alpha \}\cap \{ |\bar{v}|<\alpha  \},
  \quad
  V_{a}=V_{a}^{0}\cup V_{a}^{1}.
\end{equation*}
Note that $\xi_{a}$ is piecewise affine on a partition with constant
mesh size $\delta$ but that this partition corresponds to $P_{\delta}$
only if $a$ is a multiple of $\delta$.
For any $a$, we define $\bar{E}_{\delta}[\xi_{a}]$ by the top formula on the
right-hand side of \eqref{ee11}.
Hence
$\bar{E}_\delta[\xi_{a}] = E_\delta[\xi] \leq C$
for all values of $a$ and $E_\delta[\xi_{a}]=\bar{E}_\delta[\xi_{a}]$
if $a$ is a multiple of $\delta$.
Also note that
$\|\xi_{a}-\bar{v}\|_{2}<\infty$
follows easily from the corresponding inequality for $\xi$.
In the first part of the proof, we show that there is an $\bar{a}$
that is a multiple of $\delta$
such that $|V_{\bar{a}}|\leq 1$. 

{\bf Step 1.} We show that for any $a$, either
$V_{a}^{0}=\emptyset$ or $V_{a}^{1}=\emptyset$.
Suppose $x_{0}\in V_{a}^{0}$ and $x_{1}\in V_{a}^{1}$.  Because
$\bar{v}$ is increasing, we know that 
$x_{1}<x_{0}$.  Because $\|\xi_{a}-\bar{v}\|_{2}<\infty$ and
$\bar{v}\rightarrow 1$ as $x\rightarrow \infty$,
there is a number $x_{1}'>x_{0}$ such that
$|1-\xi_{a}(x_{1}')|<\alpha$.  
Set $\alpha_{0}=\xi_{a}(x_{0}), \alpha_{1}=1-\xi_{a}(x_{1})$,
and $\alpha_{1}'=1-\xi_{a}(x_{1}')$.  We have the estimate
\begin{align}
  \operatorname{Var}(p(\xi_{a}(x)),\reals)
  &\geq
  |p(\xi_{a}(x_{0}))-p(\xi_{a}(x_{1}))|
  +
  |p(\xi_{a}(x_{1}'))-p(\xi_{a}(x_{0}))| \nonumber \\
  &=
  |p(\alpha_{0})-p(1-\alpha_{1})|
  +
  |p(1-\alpha_{1}')-p(\alpha_{0})| \nonumber \\
  &=
  2\int_{\alpha_{0}}^{1-\alpha_{1}}\!\sqrt{w(t)} \,dt
  +
  2\int_{\alpha_{0}}^{1-\alpha_{1}'}\!\sqrt{w(t)} \,dt
  \geq
  2I_{\alpha}.
  \label{ee200}
\end{align}
This estimate and \eqref{ee196} (with $\bar{E}_{\delta}$ replacing
$E_{\delta}$) imply that $C\geq 2I_{\alpha}$, which contradicts 
our hypothesis that $I_{\alpha}> 2C/3$.  We conclude that for any
$a\in \reals$, either $V_{a}^{0}=\emptyset$ or $V_{a}^{1}=\emptyset$.

{\bf Step 2.}  We show next that the function $a\rightarrow |V_{a}^{0}|$ is
continuous.  (That $a\rightarrow |V_{a}^{1}|$ is continuous follows
by a similar argument.)
Let $\bar{x}$ be the unique
point such that $\bar{v}(\bar{x})=1-\alpha$.  Then 
$V_{a}^{0}=\{ |\xi_{a}|<\alpha \}\cap (\bar{x},\infty)$, and, because 
$\{ |\xi_{a}|<\alpha \}=\{ |\xi|<\alpha \}-a$, we see that
$V_{a}^{0}=\left(\{ |\xi|<\alpha \}\cap (\bar{x}+a,\infty)\right)-a$, hence
\begin{equation}
  |V_{a}^{0}|
  =
  |\{ |\xi|<\alpha \}\cap (\bar{x}+a,\infty)|.
  \label{ee198}
\end{equation}
%
Now suppose $a>a'>0$.  Then
\begin{equation*}
  \{ |\xi|<\alpha \}\cap (\bar{x}+a',\infty)
  =
  \left(\{ |\xi|<\alpha \}\cap (\bar{x}+a',\bar{x}+a]\right)
  \cup
  \left(\{ |\xi|<\alpha \}\cap (\bar{x}+a,\infty)\right).
\end{equation*}
Taking the measure of both sides and using \eqref{ee198} twice
implies that $0\leq |V_{a'}^{0}|-|V_{a}^{0}|\leq a-a'$, from
which continuity follows.

{\bf Step 3.} We now show that there exists an $\bar a$ such that
$\bar a$ is a multiple of $\delta$ and
$|V_{{\bar a}}|\leq 1$.  If $|V_{0}^{0}|=|V_{0}^{1}|=0$, we choose
$\bar{a}=0$.  Suppose instead that $|V_{0}^{0}|>0$ and let $x_{0}\in
V_{0}^{0}$.  (If $|V_{0}^{1}|>0$, the argument is
similar.)  Pick $x_{1}>x_{0}$ such that $|1-\xi(x_{1})|<\alpha$.  Then
for $a$ sufficiently large, $x_{1}-a\in V_{a}^{1}$ and hence
$V_{a}^{0}$ is empty.  Set $\bar{a}'=\inf\{ a>0\, :\, |V_{a}^{0}|=0 \}$.  Note
that
$\bar{a}'>0$ and that
$|V_{\bar{a}'}^{0}|=0$ by continuity.
If $|V_{\bar{a}'}^{1}|>0$, then by continuity there is a positive number $a'<\bar{a}'$ such that
$|V_{a'}^{1}|>0$.
Hence $V_{a'}^{0}$ must be empty, which contradicts the
definition of $\bar{a}'$.  Therefore $|V_{\bar{a}'}^{1}|=0$ and
thus $|V_{\bar{a}'}|=0$.
%
If $\bar{a}'$ is a multiple of $\delta$, set $\bar{a}=\bar{a}'$.
Otherwise, let $\bar{a}$ be the largest multiple of $\delta$ that is
strictly less than $\bar{a}'$.  If $|V_{\bar{a}}^{1}|>0$, then
$V_{\bar{a}}^{0}$ must be empty, which again contradicts the 
definition of $\bar{a}'$.  Hence $|V_{\bar{a}}^{1}|=0$.
Lastly, by an argument similar to that in Step~2 above, 
$|V_{\bar{a}}^{0}|\leq |V_{\bar{a}}^{0}|-|V_{\bar{a}'}^{0}|\leq \bar{a}'-\bar{a}<
\delta\leq 1$.  It follows
that $|V_{\bar{a}}|\leq 1$.

%

{\bf Step 4.} Next we claim that $-1/2\leq \xi_{\bar{a}}(x)\leq 3/2$ for all $x\in \reals$.
Suppose $\xi_{\bar{a}}(\bar{x})>3/2$ for some $\bar{x}\in \reals$.  Pick
$x_{0}<\bar{x}$ and $x_{1}>\bar{x}$ such that
$|\xi_{\bar{a}}(x_{0})|<\alpha$ and $|1-\xi_{\bar{a}}(x_{1})|<\alpha$.  Set 
$\alpha_{0}=\xi_{\bar{a}}(x_{0})$, $\bar{\alpha}=\xi_{\bar{a}}(\bar{x})$, and
$\alpha_{1}=1-\xi_{\bar{a}}(x_{1})$.  
Using \eqref{ee196} and an estimate 
similar to \eqref{ee200},
one can show that $C\geq 3I_{\alpha}/2$, which contradicts our
assumption about $C$.  
We conclude that
$\xi_{\bar{a}}(x)\leq 3/2$ for all $x\in \reals$.
A similar argument shows that $\xi_{\bar{a}}(x)\geq -1/2$ for all $x\in \reals$.

{\bf Step 5.}  Our final step is to show that there is a $\bar{C}>0$ independent of
$\xi$ such that \eqref{ee195} holds.
We define $\beta$ by the condition that $w(\xi)<\beta$ if and only if
$|\xi-m|<\alpha$ for some $m\in \mathbb{Z}$.  We decompose $\reals$ as
\begin{equation}
  \reals
  =
  (\{ w(\xi_{\bar{a}})\geq \beta \}\cup \{ w(\bar{v})\geq \beta \})
  \cup
  (\{ w(\xi_{\bar{a}})< \beta \}\cap \{ w(\bar{v})< \beta \}).
  \label{ee203}
\end{equation}
We finish the proof by estimating the $L^{2}$ norm of
$\xi_{\bar{a}}-\bar{v}$ on the sets on the right-hand side of
\eqref{ee203}.

First, for the set
$\{ w(\xi_{\bar{a}})\geq \beta \}\cup \{ w(\bar{v})\geq \beta \}$,
we have 
\begin{equation*}
  |\{ w(\xi_{\bar{a}})\geq \beta \}|
  \leq
  \beta^{-1}\int_{-\infty}^{\infty}\!w(\xi_{\bar{a}}) 
  =
  \beta^{-1}\int_{-\infty}^{\infty}\!w(\xi)
  \leq
  \beta^{-1}C
\end{equation*}
and, likewise, we have
$|\{ w(\bar{v})\geq \beta \}|\leq \beta^{-1}\tilde{C}$,
where $\tilde{C}$ does not depend on $\xi$.
Because $\xi_{\bar{a}}$ and $\bar{v}$ are both bounded, it follows
that the $L^{2}$ norm of $\xi_{\bar{a}}-\bar{v}$ over $\{
w(\xi_{\bar{a}})\geq \beta \}\cup \{ w(\bar{v})\geq \beta \}$ is
bounded by a constant that depends only on $C$.

Returning to the right-hand side of
\eqref{ee203}, we next note that
\begin{align}
  \{ w(\xi_{\bar{a}})< \beta \}\cap \{ w(\bar{v})< \beta \}
  &=
  (\{ |\xi_{\bar{a}}|<\alpha \}\cap \{ |\bar{v}|<\alpha  \}) 
  \cup
  (\{ |1-\xi_{\bar{a}}|<\alpha \}\cap \{ |1-\bar{v}|<\alpha  \})
  \cup
  V_{\bar{a}}.
  \label{ee204}
\end{align}
By Steps 3 and 4 above,
$\|\xi_{\bar{a}}-\bar{v}\|_{L^{2}(V_{\bar{a}})}$
is bounded by a constant that depends only on $C$.
We set $U_{0}=\{ |\xi_{\bar{a}}|<\alpha \}\cap \{ |\bar{v}|<\alpha  \}$
and
$U_{1}=\{ |1-\xi_{\bar{a}}|<\alpha \}\cap \{ |1-\bar{v}|<\alpha \}$.
%
Our assumptions on $w$ imply that there is a constant $\bar{w}$ such
that $w(z)\geq\bar{w}(z-m)^{2}$ if $|z-m|<\alpha$ for some
$m\in\mathbb{Z}$.  Using this, we have
\begin{align*}
  \|\xi_{\bar{a}}-\bar{v}\|_{L^{2}(U_{0})}
  \leq
  \|\xi_{\bar{a}}\|_{L^{2}(U_{0})}
  +
  \|\bar{v}\|_{L^{2}(U_{0})}
  \leq
  \bar{w}^{-1/2}\Biggl(\underset{U_{0}}{\int}\!w(\xi_{\bar{a}})\Biggr)^{1/2}
  +
  \bar{w}^{-1/2}\Biggl(\underset{U_{0}}{\int}\!w(\bar{v})\Biggr)^{1/2},  
\end{align*}
and therefore $\|\xi_{\bar{a}}-\bar{v}\|_{L^{2}(U_{0})}$ is bounded by
a constant that depends only on $C$.  The estimate on
\mbox{$\|\xi_{\bar{a}}-\bar{v}\|_{L^{2}(U_{1})}$} is similar.

We have shown that 
the $L^{2}$ norm of $\xi_{\bar{a}}-\bar{v}$
on each of the sets in the decomposition of $\reals$ given by   
\eqref{ee203} and \eqref{ee204} is bounded by a constant that depends
only on $C$ and $\delta$.
This completes Step 5.

\hfill $\Box$


Note that 
the energy $E_{\delta}[\xi]$ of $\xi\in A_{\delta}$ is invariant if
$\xi$ is horizontally translated by a multiple of $\delta$.  
When considering compactness, one could construct a sequence of
functions each of which is a translation of a given function such that
the sizes of the translations go to $\infty$.   This sequence would have
uniformly bounded energy but would fail to have a convergent
subsequence in $A$.  This observation tells us that, as a consequence
of translation invariance,  we cannot prove a
standard compactness result for the energies $E_{\delta}$.  Instead, we use the previous lemma
to establish a modified compactness result.
Let $\{ \xi_{n} \}$ be a sequence of functions in $A$ and let $\{
\delta_{n} \}$ be a bounded sequence of positive numbers.
Suppose $E_{\delta_{n}}[\xi_{n}]\leq C$,
where we assume that $p(1)<C<\frac{3}{2}I_{\alpha}$.
Using Lemma~\ref{r8}, we horizontally translate each function $\xi_{n}$ by
$a_{n}$, where $a_{n}$ is the constant from the lemma.  We continue to
denote this
sequence of translated functions by $\{ \xi_{n} \}$.  
Every function in the sequence satisfies \eqref{ee195}, where $\bar{C}$
depends only on $C$.
Because $H^{1}(\reals)$ is reflexive, it follows easily
that there is a subsequence $\{ \xi_{n_{k}} \}$ and $\bar{\xi}\in A$
such that $\xi_{n_{k}}$ converges weakly to $\bar{\xi}$ in $A$.

We can interpret the previous compactness result as follows.  Each
function in the initial sequence is in the space $A$.  Hence each has
a transition region, where the displacements change from values near 0
to values near 1; see the proof of Lemma~\ref{r8}.  Because of the
translation invariance in our problem, we cannot expect to extract a
subsequence for which the locations of these transition regions
converge to a fixed position in space.  As an alternative, we can
consider an observer who is attached near the center of the transition
region of each function.  This observer sees the profiles of the
transition regions.  Our compactness result shows that this observer
can select a subsequence of functions such that the corresponding
profiles converge weakly in $A$.

We now state and prove the main result of this section, which is  a $\Gamma$-convergence result for the
energies $E_{\delta}$. 
\begin{theorem}
Let $E_{\delta}$ be defined as in \eqref{ee11} and let $F$ be defined
as in \eqref{ee7}.  Let $\delta_{n}\rightarrow 0^{+}$ as
$n\rightarrow\infty$.  Then $E_{\delta_{n}}$ $\Gamma$-converges to $F$
in the weak topology of $A$.  \label{r1}
\end{theorem}  

\noindent \textit{Proof}.
We first prove the liminf inequality.  Let $\xi^{*}\in A$ and
consider a sequence $\{ \xi_{n} \}\subset A$ that converges to
$\xi^{*}$ weakly in $A$.  For the elastic term in $E_{\delta}$, we
note that because $\{\xi_{n} \}$ converges weakly to $\xi^{*}$ in $A$,
$\{\xi_{n}' \}$ converges weakly to $(\xi^{*})'$ in $L^{2}(\reals)$.
Hence
\begin{equation}
  \liminf_{n\rightarrow \infty} \int_{-\infty}^\infty \xi_{n}^{\prime}(x)^{2}\,dx
  \geq
  \int_{-\infty}^\infty  (\xi^{*})^{\prime}(x)^2\,dx. \label{ee16}
\end{equation}
by the weak lower semicontinuity of the $L^{2}$ norm.

Next we show that
\begin{equation}
  \int_{-\infty}^{\infty}\!w\bigl(\xi_{n}(x)\bigr) \,dx
  \rightarrow
  \int_{-\infty}^{\infty}\!w\bigl(\xi^{*}(x)\bigr) \,dx
  \quad
  \text{as}
  \quad
  n\rightarrow \infty, 
  \label{ee27}
\end{equation}
which combined with \eqref{ee16} establishes the liminf inequality.
To estimate
$
  \left|
  \int_{0}^{\infty}\!w\bigl(\xi_{n}\bigr) 
  -
  \int_{0}^{\infty}\!w\bigl(\xi^{*}\bigr) 
  \right|,
$
we expand
\begin{equation}
  w\bigl(\xi_{n}(x)\bigr)
  -
  w\bigl(\xi^{*}(x)\bigr)
  =
  w'\bigl(\xi^{*}(x)\bigr)
  \bigl(\xi_{n}(x)-\xi^{*}(x)\bigr)
  +
  \frac{1}{2}w''\bigl(\xi_{x}\bigr)
  \bigl(\xi_{n}(x)-\xi^{*}(x)\bigr)^{2},
  \label{ee30}
\end{equation}
where $\xi_{x}$ is a number between $\xi^{*}(x)$ and $\xi_{n}(x)$.
For the first term on the right-hand side of \eqref{ee30}, we have
\begin{equation*}
  \left|
  \int_{0}^{\infty}\! 
  w'\bigl(\xi^{*}(x)\bigr)
  \bigl(\xi_{n}(x)-\xi^{*}(x)\bigr) \,dx
  \right|^{2}
  \leq
  \int_{0}^{\infty}\! 
  w'\bigl(\xi^{*}(x)\bigr)^{2}\,dx
  \int_{0}^{\infty}\! 
  \bigl(\xi_{n}(x)-\xi^{*}(x)\bigr)^{2} \,dx.  
\end{equation*}
We next expand
\begin{equation}
  w'\bigl(\xi^{*}(x)\bigr)
  =
  w''\bigl(\hat{\xi}_{x}\bigr)
  \bigl(\xi^{*}(x)-1\bigr),
  \label{ee32}
\end{equation}
where $\hat{\xi}_{x}$ is between $\xi^{*}(x)$ and 1 and where we use that 
$w'\bigl(1\bigr)=0$.  
By the triangle inequality, we have
\begin{equation}
  \left[
  \int_{0}^{\infty}\! 
  \bigl(\xi^{*}(x)-1\bigr)^{2}\,dx
  \right]^{1/2}
  \leq
  \left[
  \int_{0}^{\infty}\! 
  \bigl(\xi^{*}(x)-\bar{v}(x)\bigr)^{2}\,dx
  \right]^{1/2}
  +
  \left[
  \int_{0}^{\infty}\! 
  \bigl(\bar{v}(x)-1\bigr)^{2}\,dx
  \right]^{1/2}.
  \label{ee33}
\end{equation}
The first term on the right-hand side is finite because $\xi^{*}\in
A$, and the second term is finite by our assumptions on $\bar{v}$.

From \eqref{ee30}--\eqref{ee33}, it follows that there is a
constant $C$ such that
\begin{equation*}
  \left|
  \int_{0}^{\infty}\!w\bigl(\xi_{n}(x)\bigr) \,dx
  -
  \int_{0}^{\infty}\!w\bigl(\xi^{*}(x)\bigr) \,dx
  \right|
  \leq
  C
  \Biggl[
  \int_{0}^{\infty}\! 
  \bigl(\xi_{n}(x)-\xi^{*}(x)\bigr)^{2} \,dx
  \Biggr]^{1/2}.
\end{equation*}
We can make a similar estimate for 
\begin{equation*}
  \left|
  \int_{-\infty}^{0}\!w\bigl(\xi_{n}(x)\bigr) \,dx
  -
  \int_{-\infty}^{0}\!w\bigl(\xi^{*}(x)\bigr) \,dx
  \right|.
\end{equation*}
(In this case the right-hand side of \eqref{ee32} is 
$w''\bigl(\hat{\xi}_{x}\bigr) \xi^{*}(x)$ and the second term on the
right-hand side of \eqref{ee33} is $\int_{-\infty}^{0}\!\bar{v}^{2} $.)
Combining these estimates and using that
$\|\xi_{n}-\xi^*\|_{2}\rightarrow 0$ 
establishes \eqref{ee27}.

Now we prove the recovery sequence argument.  Let $\xi^{*}\in A$.
Following a standard construction (see \cite{braides2002gamma}), for
$n\in \mathbb{Z}$ we define $\xi_{n}$ by setting
$\xi_{n}(i\delta_{n})=\xi^{*}(i\delta_{n})$ for $i\in \mathbb{Z}$
and requiring that
$\xi_{n}\in A_{\delta_{n}}$.  Hence $\xi_{n}$ is continuous and
piecewise affine with respect to the partition $P_{\delta_{n}}$.

We claim that $\xi_{n}$ converges to $\xi^{*}$ weakly in $A$.
First we show that
$\xi_{n}-\bar{v}$ converges to
$\xi^{*}-\bar{v}$ in $L^{2}(\reals)$.
If $x\in [(i-1)\delta_{n},i\delta_{n}]$, then
$\xi_{n}(x)$ is between $\xi^{*}((i-1)\delta_{n})$ and $\xi^{*}(i\delta_{n})$.
Hence for $i\in\mathbb{Z}$, we have 
\begin{align*}
  \int_{(i-1)\delta_{n}}^{i\delta_{n}}\!|\xi^{*}(x)&-\xi_{n}(x)|^{2} \,dx   \nonumber \\
  &\leq
  \int_{(i-1)\delta_{n}}^{i\delta_{n}}\!|\xi^{*}(x)-\xi^{*}((i-1)\delta_{n})|^{2} \,dx
  +
  \int_{(i-1)\delta_{n}}^{i\delta_{n}}\!|\xi^{*}(x)-\xi^{*}(i\delta_{n})|^{2} \,dx
  \nonumber \\
  &\leq
  \int_{(i-2)\delta_{n}}^{i\delta_{n}}\!|\xi^{*}(x)-\xi^{*}((i-1)\delta_{n})|^{2} \,dx
  +
  \int_{(i-1)\delta_{n}}^{(i+1)\delta_{n}}\!|\xi^{*}(x)-\xi^{*}(i\delta_{n})|^{2} \,dx,
\end{align*}
which implies that
\begin{equation}
  \int_{-\infty}^{\infty}\!|\xi^{*}(x) - \xi_{n}(x)|^{2} \,dx
  \leq
  2 \sum_{i=-\infty}^{\infty}
  \int_{(i-1)\delta_{n}}^{(i+1)\delta_{n}}\!|\xi^{*}(x)-\xi^{*}(i\delta_{n})|^{2} \,dx.
  \label{ee157}
\end{equation}
Estimating a typical term in the sum on the right-hand side of
\eqref{ee157} gives
\begin{align*}
  \int_{(i-1)\delta_{n}}^{(i+1)\delta_{n}}\! &
  |\xi^{*}(x)-\xi^{*}(i\delta_{n})|^{2} \,dx 
  =
  \int_{(i-1)\delta_{n}}^{(i+1)\delta_{n}}\!
  \left|
  \int_{i\delta_{n}}^{x}\!(\xi^{*})'(y) \,dy
  \right|^{2} \,dx \nonumber \\
  &\leq
  \delta_{n} \int_{(i-1)\delta_{n}}^{(i+1)\delta_{n}}\!
  \left|\int_{i\delta_{n}}^{x}\! (\xi^{*})'(y)^{2}\,dy\right|
  dx \nonumber \\
  &=
  \delta_{n}
  \left[
    \int_{(i-1)\delta_{n}}^{i\delta_{n}}\!
    \int_{(i-1)\delta_{n}}^{y}\!(\xi^{*})'(y)^{2} \,dx\,
    dy
    +
    \int_{i\delta_{n}}^{(i+1)\delta_{n}}\!
    \int_{y}^{(i+1)\delta_{n}}\!(\xi^{*})'(y)^{2} \,dx\,
    dy
  \right] \nonumber \\
  &\leq
  \delta_{n}^{2}
  \int_{(i-1)\delta_{n}}^{(i+1)\delta_{n}}\!(\xi^{*})'(y)^{2} \,dy.
\end{align*}
By using this estimate in \eqref{ee157}, we arrive at
\begin{equation}
  \int_{-\infty}^{\infty}\!|\xi^{*}(x) - \xi_{n}(x)|^{2} \,dx
  \leq
  4 \delta_{n}^{2} 
  \int_{-\infty}^{\infty}\!(\xi^{*})'(y)^{2}\,dy.
  \label{ee159}
\end{equation}
We observe that $(\xi^{*})'\in L^{2}(\reals)$ because
$(\xi^{*})'-\bar{v}'$, $\bar{v}'\in L^{2}(\reals)$.
Hence it follows from \eqref{ee159} that 
$\xi_{n}-\bar{v}$ converges to
$\xi^{*}-\bar{v}$ in $L^{2}(\reals)$.


It remains to show that $\xi_{n}'-\bar{v}'$ converges to
$(\xi^{*})'-\bar{v}'$ weakly in $L^{2}(\reals)$.
Let $g\in L^{2}(\reals)$ and let $\eta>0$.   Pick
$h\in C_{0}^{\infty}(\reals)$ such that $\|g-h\|_{2}<\eta$.
We have
\begin{align}
  \left|
  \int_{-\infty}^{\infty}\![(\xi^{*})'- \xi_{n}']g 
  \right|
  &=
  \left|
  \int_{-\infty}^{\infty}\![(\xi^{*})'- \xi_{n}']h
  +
  \int_{-\infty}^{\infty}\![(\xi^{*})'- \xi_{n}'](g-h)  
  \right| \nonumber \\
  &\leq
  \|\xi^{*}- \xi_{n}\|_{2} \|h'\|_{2}
  +
  \|(\xi^{*})'- \xi_{n}'\|_{2} \|g-h\|_{2} \nonumber \\
  &\leq
  \|\xi^{*}- \xi_{n}\|_{2} \|h'\|_{2}
  +
  2\|(\xi^{*})'\|_{2}\eta.
\label{ee161}
\end{align}
Note that the final inequality uses \eqref{ee47}, which is established
below.  The first term on the right-hand side of \eqref{ee161} goes to
0 as $n\rightarrow \infty$.  Since $\eta>0$ is arbitrary, we
conclude that $\underset{n\rightarrow \infty}{\lim}
\int_{-\infty}^{\infty}[(\xi^{*})'- \xi_{n}']g=0$.

Now we consider $\limsup E_{\delta_{n}}[\xi_{n}]$.   
We use an easy argument
from \cite{braides2002gamma} for the elastic energy.
For each $i\in \mathbb{Z}$, we have
\begin{align}
  \int_{i\delta_{n}}^{(i+1)\delta_{n}}\!(\xi^{*})'(x)^{2}\,dx
  &=
  \frac{1}{\delta_{n}}\int_{i\delta_{n}}^{(i+1)\delta_{n}}\!(\delta_{n}^{1/2}(\xi^{*})'(x))^{2}\,dx
  \nonumber\\  
  &\geq
  \biggl(
    \delta_{n}^{-1/2}\int_{i\delta_{n}}^{(i+1)\delta_{n}}\!(\xi^{*})'(x)\,dx
  \biggr)^{2}\nonumber \\[2mm]
  &=\delta_{n}^{-1}\bigl(\xi^{*}((i+1)\delta_{n})-\xi^{*}(i\delta_{n})\bigr)^{2}
  \label{ee48}\\[2mm]
  &=\delta_{n}^{-1}\bigl(\xi_{n}((i+1)\delta_{n})-\xi_{n}(i\delta_{n})\bigr)^{2}
  \nonumber \\[2mm]
  &=\delta_{n}\left(\frac{\xi_{n}((i+1)\delta_{n})-\xi_{n}(i\delta_{n})}{\delta_{n}}\right)^{2}
  %
  =
  \int_{i\delta_{n}}^{(i+1)\delta_{n}}\! \xi_{n}'(x)^{2}\,dx.    
  \nonumber 
\end{align}
Note that the second line in \eqref{ee48} uses Jensen's Inequality.
Summing over all $i\in \mathbb{Z}$ gives
\begin{equation}
  \int_{-\infty}^{\infty}\! (\xi^{*})'(x)^{2}\,dx
  \geq
  \int_{-\infty}^{\infty}\! \xi_{n}'(x)^{2}\,dx.
  \label{ee47}
\end{equation}

For the weak interaction energy, it is sufficient to observe that the
recovery sequence $\{ \xi_{n}\}$ constucted above converges weakly to
$\xi^{*}$ in $A$.  Hence the same proof we used above for the liminf
inequality gives us \eqref{ee27}.

\hfill $\Box$

\section{Convergence to a Macroscale Model} \label{s5}

As noted in the introduction, we can view the limiting process in which
$\delta\rightarrow 0$ as
taking the system
from the atomic or microscale to a mesoscale associated with a single
domain wall.
In this section, we introduce another small parameter into the
discrete energy \eqref{ee6}, and we prove a convergence result for
this modified version of the discrete energies.
The associated limiting process of letting this additional small
parameter go to $0$ takes the system from the mesoscale to a
macroscale associated with multiple domain walls.

To motivate this new rescaling and the related limiting process, we let
$\varepsilon>0$ be a small parameter and make the change of variables
$\hat{x}=\varepsilon x$ in the limiting energy \eqref{ee7}.  This
yields
\begin{align}
  \int_{-\infty}^{\infty}\!
  &
   \xi'(\hat{x}/\varepsilon)^{2} \varepsilon^{-1} \,d\hat{x}
  +
  \int_{-\infty}^{\infty}\!
  w(\xi(\hat{x}/\varepsilon)) \varepsilon^{-1} \,d\hat{x} \nonumber \\
  &=
  \varepsilon 
  \int_{-\infty}^{\infty}\!
   \hat{\xi}'(\hat{x})^{2}  \,d\hat{x}
  +
  \varepsilon^{-1}\int_{-\infty}^{\infty}\!
  w(\hat{\xi}(\hat{x})) \,d\hat{x},  
  \label{ee154}
\end{align}
where $\hat{\xi}(\hat{x})\leqr\xi(\hat{x}/\varepsilon)$.  Based upon this,
we define the energy $F_{\varepsilon}$ by
\begin{equation}
  F_{\varepsilon}[\xi]
  \leqr
  \varepsilon\int_{-\infty}^{\infty}\! \xi'(x)^{2} \,dx
  +
  \varepsilon^{-1}\int_{-\infty}^{\infty}\!w(\xi(x)) \,dx.
  \label{ee18}
\end{equation}
Based on standard results in the literature, we expect that 
$\{F_{\varepsilon} \}$ $\Gamma$-converges as $\varepsilon\rightarrow 0$.
In Section~\ref{s3}, we
showed that \eqref{ee7} is the $\Gamma$-limit of the energy
\eqref{ee11} as $\delta\rightarrow 0$.  Recall that the energy
\eqref{ee11} is a continuum version of the discrete energy
\eqref{ee6}.
%
%
These observations suggest that introducing
$\varepsilon$ into the discrete energy \eqref{ee6} in a way analogous
to the rescaling that takes us from \eqref{ee7} to \eqref{ee18}
will give a discrete energy depending on the 2 small
parameters $\delta$ and $\varepsilon$ that should $\Gamma$-converge 
%
%
as these parameters go to $0$.  In this section,
we formulate and prove this result.

As mentioned, introducing $\varepsilon$ and allowing
$\varepsilon\rightarrow 0$ takes the system from the mesoscopic to
the macroscopic scale.
The scalings we have considered provide two ways of carrying the
system from the microscale to the macroscale.
On the one hand, we can go from micro to macro in two steps.  The first step is
from the macroscale to a mesoscale by the result presented in the
previous section.  The second
step is from this mesoscale to the macroscale by intoducing $\varepsilon$ in the limiting
continuum energy \eqref{ee7} from the first part of this paper and then letting
$\varepsilon\rightarrow 0$.

On the other hand, a second way to go directly from micro to macro 
is by introducing
$\varepsilon$ in the discrete energies \eqref{ee6} and letting $\varepsilon$,
$\delta$ both $\rightarrow 0$.
To rewrite the discrete energy appropriately, we multiply
and divide the first term in \eqref{ee6} by $\varepsilon^{2}$ and
the second term in \eqref{ee6} by $\varepsilon$.  This yields
\begin{equation}
  E^{\varepsilon,\delta}[\{ \xi_{i}  \}]
  =
  \varepsilon\sum_{i=-\infty}^{\infty}
  \left(\frac{\xi_{i}-\xi_{i-1}}{\varepsilon\delta}\right)^{2}\varepsilon\delta
    +
  \frac{1}{\varepsilon}\sum_{i=-\infty}^{\infty}
    w(\xi_{i})\varepsilon\delta,
  \label{ee118b}
\end{equation}
where $w$ is defined in \eqref{ee8}.
We think of \eqref{ee118b} as the discrete version of, or a Riemann
sum for, the continuum energy \eqref{ee18}.  Likewise, we think of
\eqref{ee6} as a discrete version of \eqref{ee11}.  From this
perspective, $\delta$ in \eqref{ee6} is like $dx$ in \eqref{ee11} and
$\varepsilon\delta$ in \eqref{ee118b} is like $dx$ in \eqref{ee18}.
This is consistent with the relation between $dx$ in \eqref{ee7} and
$dx$ in \eqref{ee18}, per the derivation in \eqref{ee154}.

Next we define the appropriate spaces for the convergence result of
this section.  In the previous section, we assumed that
the displacement went from $0$ at $-\infty$ to $1$ at $\infty$.  Now
we wish to consider more general displacement boundary conditions.  We
pick $m_{l}, m_{r}\in \mathbb{Z}$, where $m_l\neq m_r$.  We let
$\bar{\bar{v}}$ be a continuously differentiable, strictly monotone
function satisfying 
\begin{gather}
  \bar{\bar{v}}(x)
  \rightarrow 
  \begin{Bmatrix}
   m_{r}\\
   m_{l}
  \end{Bmatrix}
  \ \text{as}\ 
  x\rightarrow 
  \begin{Bmatrix}
   \infty\\
   -\infty
  \end{Bmatrix}, \label{ee2011}
  \\[2mm]
  1-\bar{\bar{v}}\in L^{1}([0,\infty))\cap L^{2}([0,\infty)),
  \
  \bar{\bar{v}}\in L^{1}((-\infty,0])\cap L^{2}((-\infty,0]),
  \
  \bar{\bar{v}}'\in L^{1}(\reals)\cap L^{2}(\reals).
\label{ee2001}
\end{gather}
%
%
For example, we could
choose 
\begin{equation*}
    \bar{\bar{v}}=(m_{r}-m_{l})\bar{v}+m_{l},
\end{equation*}
where $\bar{v}$ is defined
as in \eqref{ee21}.  
Let $\bar{C}$ be a positive constant.
We define the
space
\begin{equation}
  \bar{A}
  :=
  \{ \xi: \xi-\bar{\bar{v}}\in L^{1}(\reals)\
  \text{and}\ 
  \|\xi-\bar{\bar{v}}\|_{L^{1}(\reals)} \leq \bar{C} \}.
  \label{ee112}
\end{equation}
Note that $\bar{A}\subset L^{1}_{\text{loc}}(\reals)$.
We say that the sequence 
$\{\xi_{n}\}$ converges in $\bar{A}$ if 
$\xi_{n}\rightarrow \xi$ in $L^{1}_{\text{loc}}(\reals)$.

We comment on the condition that
$\|\xi-\bar{\bar{v}}\|_{L^{1}(\reals)} \leq \bar{C}$ in the definition
of $\bar{A}$.  This condition is needed because under the more general
assumption \eqref{ee2011} on the limiting behavior of $\bar{\bar{v}}$,
we cannot prove a result analogous to Lemma~\ref{r8}.
%
%
For example, suppose $m_{l}=0$ and $m_{r}=2$ and let $\xi$ be a smooth
increasing function that approaches $0$ and $2$ as $x \rightarrow
-\infty$ and $x \rightarrow \infty$.  Further suppose that $\xi(x)=1$
for $0\leq x\leq b$ with a sharp transition from $0$ to $1$ to the
left of $0$ and a sharp transition from $1$ to $2$ to the right of
$b$.
The right endpoint $b$ can be increased without increasing the energy
$E_{\varepsilon,\delta}[\xi]$.  However, increasing $b$ increases the
minimum over $a$ of the
$L^{1}$ distance between $\bar{\bar{v}}$ and $\xi_{a}$, where
$\xi_{a}$ is the horizontal
translation of $\xi$ by $a$.

Next, define
\begin{equation*}
  \bar{A}_{\rho}
  :=
  \{ \xi: \xi\in \bar{A},\  \xi\ \text{continuous, $\xi$ piecewise affine on $P_{\rho}$}\}
\end{equation*}
for $\rho>0$, where $P_{\rho}$ is the partition defined as in
\eqref{ee20}. 

For any $\varepsilon$, $\delta>0$, we define
$E_{\varepsilon,\delta}:\bar{A}\rightarrow [0,\infty]$
by
\begin{equation}
  E_{\varepsilon,\delta}[\xi]
  :=
  \left\{
  \begin{array}{ll}
    \varepsilon\int_{-\infty}^{\infty}\xi^\prime(x)^{2}\,dx
    +\varepsilon^{-1}\int_{-\infty}^{\infty}w\left(\xi(x)\right)\,dx, & \xi\in A_{\varepsilon\delta}, \\
    \infty, & \xi\in \bar{A}\setminus A_{\varepsilon\delta}.
  \end{array}
  \right. \label{ee60}
\end{equation}
Note that here, as in \eqref{ee11}, we have replaced the discrete
interaction energy in \eqref{ee118b} with a corresponding integral
expression.  By a lemma similar to Lemma~\ref{l:0}, we could prove here that
the difference in the energy
\eqref{ee118b} and the energy on the top line of
\eqref{ee60}
goes to $0$ as $\epsilon$, $\delta \rightarrow 0$.
Next we define
\begin{equation}
  E[\xi]
  :=
  \left\{
  \begin{array}{ll}
    \operatorname{Var}(\xi,\reals)\bar{p},
    & \xi\in BV(\reals;\mathbb{Z})\cap \bar{A},\\ 
    \infty, & \xi\in \bar{A}\setminus BV(\reals;\mathbb{Z}).
  \end{array}
  \right. \label{ee61}
\end{equation}
Here $\bar{p}=p(1)$ where $p$ is defined in \eqref{ee194}.


We establish a compactness result in $\bar{A}$.
Let $\varepsilon_{n}, \delta_{n}\rightarrow 0^{+}$ as $n\rightarrow
\infty$. 
Let $\{\xi_{n}\}\subset \bar{A}$ such that 
\begin{equation*}
  E_{\varepsilon_{n},\delta_{n}}[\xi_{n}]\leq C,
\end{equation*}
where $C$ is a positive constant.
We show that
$\{\xi_{n}\}$
has a subsequence that converges in $\bar{A}$ to a function
$\xi^{*}\in BV(\reals;\mathbb{Z})\cap\bar{A}$.

%
The first step is to apply a standard BV compactness result to the
sequence
$\{p(\xi_{n})-p(\bar{\bar{v}})\}$.  
Note that $p'$ is bounded on $\reals$ and hence there is a constant
$M$ such that
$|p(\xi_{n}(x))-p(\bar{\bar{v}}(x))|
  \leq
  M|\xi_{n}(x)-\bar{\bar{v}}(x)|$.
Using this and \eqref{ee112}, we see that for any $n$
\begin{equation}
  \|p(\xi_{n})-p(\bar{\bar{v}})\|_{L^{1}(\reals)} \leq M\bar{C}.
  \label{ee71}
\end{equation}

We next seek an estimate like \eqref{ee71} but on the derivative of
$p(\xi_{n})-p(\bar{\bar{v}})$.
We have
\begin{align*}
  \int_{-\infty}^{\infty}\!
  \left|
  \frac{d}{dx}p(\xi_{n})
  \right|\,dx
  &=
  \int_{-\infty}^{\infty}\!
  \left|
  2\sqrt{w(\xi_{n})}\xi'_{n}
  \right|\,dx \nonumber \\
  &\leq
  \varepsilon_{n}\int_{-\infty}^{\infty}\xi_{n}^\prime(x)^{2}\,dx
  +
  \varepsilon_{n}^{-1}\int_{-\infty}^{\infty}w\left(\xi_{n}(x)\right)\,dx
  \nonumber \\
  &\leq
  C.
\end{align*}
Hence
\begin{equation*}
  \left\|\frac{d}{dx}\{p(\xi_{n})-p(\bar{\bar{v}})\}\right\|_{L^{1}(\reals)}
  \leq
  \left\|\frac{d}{dx}p(\xi_{n})\right\|_{L^{1}(\reals)}
  +
  \left\|\frac{d}{dx}p(\bar{\bar{v}})\right\|_{L^{1}(\reals)}
  \leq C  +
  M\left\|\bar{\bar{v}}'\right\|_{L^{1}(\reals)}.
\end{equation*}
A standard compactness result for BV functions
\cite{gariepy2001functions}
implies that there is a subsequence $\{\xi_{n_{k}}\}$ and a
function $h^{*}\in BV(\reals)$ such that
$p(\xi_{n_{k}})-p(\bar{\bar{v}})\rightarrow h^{*}$
in
$L^{1}_{\text{loc}}(\reals)$ as $k\rightarrow \infty$.
Hence
$p(\xi_{n_{k}})\rightarrow g^{*}\leqr h^{*}+p(\bar{\bar{v}})$
in $L^{1}_{\text{loc}}(\reals)$.
An easy argument shows that by
passing to a subsequence we can assume that
$p(\xi_{n_{k}})\rightarrow g^{*}$ a.e\@. on $\reals$.  Because $p$ is
one-to-one, there is a $\xi^{*}(x)$
such that $g^{*}(x)=p(\xi^{*}(x))$.  Therefore
$p(\xi_{n_{k}})\rightarrow p(\xi^{*})$, which implies that 
$\xi_{n_{k}}\rightarrow \xi^{*}$ a.e. on $\reals$.

We show that $\xi_{n_{k}}\rightarrow \xi^{*}$ in
$L^{1}_{\text{loc}}(\reals)$.  Because $p'$ is periodic, $p'\geq 0$,
and $p'=0$ only at isolated points, it follows that for any $\mu>0$,
there is a $\gamma>0$ such that
\begin{equation*}
  |p(y)-p(y')|<\gamma|y-y'|
  \quad
  \text{implies that}
  \quad
  |y-y'|<\mu.
\end{equation*}
Let $D$ be a compact set in $\reals$.  Given $k\in \mathbb{N}$, define
\begin{equation*}
  A_{k}
  =
  \{ x\in D
  \ :\
  |p(\xi_{n_{k}}
(x))-p(\xi^{*}(x))|
  <
  \gamma|\xi_{n_{k}}
(x)-\xi^{*}(x)|  \}.
\end{equation*}
Then
\begin{align*}
  \int_{D}\! |\xi_{n_{k}}
(x)-\xi^{*}(x)|\,dx
  &=
  \int_{A_{k}}\! |\xi_{n_{k}}
(x)-\xi^{*}(x)|\,dx
  +
  \int_{D-A_{k}}\! |\xi_{n_{k}}
(x)-\xi^{*}(x)|\,dx \nonumber\\[2mm]
  &\leq
  \int_{A_{k}}\! \mu\,dx
  +
  \int_{D-A_{k}}\! \gamma^{-1}|p(\xi_{n_{k}}
(x))-p(\xi^{*}(x))|\,dx 
\\[2mm]
  &\leq
  \mu|D|+\gamma^{-1}\|p(\xi_{n_{k}}
)-p(\xi^{*})\|_{L^{1}(D)},
  \nonumber  
\end{align*}
which implies that 
\begin{equation*}
  \limsup \int_{D}\! |\xi_{n_{k}}
(x)-\xi^{*}(x)|\,dx
  \leq
  \mu|D|.
\end{equation*}
Because $\mu>0$ is arbitrary, it follows that 
$\|\xi_{n_{k}}
-\xi^{*}\|_{L^{1}(D)}\rightarrow 0$.  

We know that
$\int_{-\infty}^{\infty}\!w(\xi_{n_{k}}) \,dx\leq
\epsilon_{n_{k}}C\rightarrow 0$
as $k\rightarrow \infty$.  By passing to another subsequence, we can
assume that $w(\xi_{n_{k}})\rightarrow 0$ a.e.\@ on $\reals$.  It
follows that $w(\xi^{*})=0$ a.e.\@ and hence $\xi^{*}(x)\in
\mathbb{Z}$ a.e.  
Furthermore, 
because $g^{*}\in BV(\reals;p\left(\mathbb{Z})\right)$, 
there are numbers $t_{1}<\cdots<t_{n}$ and integers $z_{i}$ such
that 
$g^{*}=\sum_{i=1}^{n+1}p(z_{i})\chi_{(t_{i-1},t_{i})}$ with
$t_{0}=-\infty$ and $t_{n+1}=\infty$.
Therefore $\xi^{*}=\sum_{i=1}^{n+1}z_{i}\chi_{(t_{i-1},t_{i})}$, which
tells us that $\xi^{*}\in BV(\reals;\mathbb{Z})$.  

It remains to verify that $\xi^{*}\in \bar{A}$.  
We start by showing that $z_{1}=m_{l}$ and 
$z_{n+1}=m_{r}$.
We can use local $L^{1}$ convergence to construct a
subsequence of $\{ \xi_{n_{k}}  \}$ (still denoted by
$\{ \xi_{n_{k}}\}$) where, for each $k\in \mathbb{N}$, we pick
$n_{k}$ such that
$\int_{t_{n}}^{t_{n}+k}\!|\xi_{n_{k}}-\xi^{*}| < 1/k$.
Then
\begin{align*}
  \bar{C}
  \geq
  \|\xi_{n_{k}}-\bar{\bar{v}}\|_{L^{1}(\reals)}
  \geq
  \int_{t_{n}}^{t_{n}+k}\!|\xi_{n_{k}}-\bar{\bar{v}}|  
  &\geq
  \int_{t_{n}}^{t_{n}+k}\!|\xi^{*}-m_{r}|
  -
  \int_{t_{n}}^{t_{n}+k}\!|\bar{\bar{v}}-m_{r}| 
  -
  \int_{t_{n}}^{t_{n}+k}\!|\xi_{n_{k}}-\xi^{*}| \nonumber \\
  &\geq
  k|z_{n+1}-m_{r}|
  -
  \int_{t_n}^{\infty}\!|\bar{\bar{v}}-m_{r}|
  -
  1/k,
\end{align*}
which by letting $k\rightarrow \infty$ implies that $m_{r}=z_{n+1}$.  Showing that
$z_{1}= m_{l}$ is similar.
Now we verify that
$\|\xi^{*}-\bar{\bar{v}}\|_{L^{1}(\reals)}\leq\bar{C}$.
For $\mu>0$, we can pick $M$ large enough such that
$\int_{-\infty}^{-M}\! |\bar{\bar{v}}-m_{l}|
+
\int_{M}^{\infty}\! |\bar{\bar{v}}-m_{r}|
<\mu
$.
Then
\begin{equation*}
  \int_{-\infty}^{\infty}\!|\xi^{*}-\bar{\bar{v}}|
  \leq
  \int_{-M}^{M}\!|\xi^{*}-\bar{\bar{v}}| + \mu
  \leq
  \int_{-M}^{M}\!|\xi^{*}-\xi_{n_{k}}|
  +
  \int_{-M}^{M}\!|\xi_{n_{k}}-\bar{\bar{v}}| + \mu    
  \leq
  \int_{-M}^{M}\!|\xi^{*}-\xi_{n_{k}}|
  +
  \bar{C} + \mu.    
\end{equation*}
The integral on the right-hand side goes to 0 as $k$ goes to infinity. 
Because $\mu>0$ is arbitrary,
$\|\xi^{*}-\bar{\bar{v}}\|_{L^{1}(\reals)}\leq\bar{C}$ follows.



Now we prove the main result of this section.
\begin{theorem}
Let $E_{\varepsilon,\delta}$ be defined as in \eqref{ee60} and let
$E$ be defined as in \eqref{ee61}.
Let $\varepsilon_{n}, \delta_{n}\rightarrow 0^{+}$ as $n\rightarrow
\infty$.  Then
$E_{\varepsilon_{n},\delta_{n}}$  $\Gamma$-converges to $E$ with
respect to local $L^{1}$ convergence in $\bar{A}$.  
\end{theorem}  

\noindent \textit{Proof}. 

{\bf Step 1.}
We start by proving the liminf inequality.
Let $\{ \xi_{n}  \}$ be a sequence such that
$\xi_{n}\rightarrow \xi^{*}$ in $\bar{A}$.  If
$\liminf E_{\varepsilon_{n},\delta_{n}}[\xi_{n}]=\infty$, there is
nothing to prove.
Hence by passing
to a subsequence, we can assume that
$\lim E_{\varepsilon_{n_{k}},\delta_{n_{k}}}[\xi_{n_{k}}]
=
\liminf E_{\varepsilon_{n},\delta_{n}}[\xi_{n}]<\infty$.
Then $\xi_{n_{k}}\in A_{\varepsilon_{n_{k}},\delta_{n_{k}}}$ for all $k$, and 
by the preceding compactness result
we can further assume that $\xi_{n_{k}}\rightarrow \xi^{*}$
pointwise a.e.\@ on $\reals$ and that 
$\xi^{*}\in BV(\reals;\mathbb{Z})$.  
To simplify notation, we write just $\{ \xi_{n} \}$ to denote the
subsequence.

Parts of the next argument follow closely the notes \cite{leoni2013gamma}.
We write $\xi^{*}=\sum_{i=1}^{n+1}z_{i}\chi_{(t_{i-1},t_{i})}$ for
numbers $t_{1}<\cdots<t_{n}$ and integers $z_{i}$, with
$t_{0}=-\infty$ and $t_{n+1}=\infty$.
Let $\gamma$ be a small positive number.  We have
\begin{align*}
  E_{\varepsilon_{n},\delta_{n}}[\xi_{n}]
  &\geq
  \sum_{i=1}^{n}
  \left[ \varepsilon_{n}\int_{t_{i}-\gamma}^{t_{i}+\gamma}\xi_{n}^\prime(x)^{2}\,dx
  +
  \varepsilon_{n}^{-1}\int_{t_{i}-\gamma}^{t_{i}+\gamma}w\left(\xi_{n}(x)\right)\,dx
  \right].
\end{align*}
We consider the $i$th term in the sum on the right-hand side.  By a
change of variables we can assume that $t_{i}=0$.  By taking $\gamma$
smaller, we can assume that $\{ \xi_{n}(-\gamma) \}$ converges to
$z_{i}$ and that $\{ \xi_{n}(\gamma) \}$ converges to $z_{i+1}$.  Then
\begin{align*}
  \varepsilon_{n}\int_{-\gamma}^{\gamma}\xi_{n}^\prime(x)^{2}\,dx
  +
  \varepsilon_{n}^{-1}\int_{-\gamma}^{\gamma}w\left(\xi_{n}(x)\right)\,dx
  &\geq
  2\int_{-\gamma}^{\gamma}\sqrt{w\left(\xi_{n}(x)\right)}\xi_{n}^\prime(x)\,dx
  \nonumber \\
  &=
  2\int_{\xi_{n}(-\gamma)}^{\xi_{n}(\gamma)}\sqrt{w\left(s\right)}\,ds
  \rightarrow
  2\int_{z_{i}}^{z_{i+1}}\sqrt{w\left(s\right)}\,ds  
\end{align*}
as $n\rightarrow \infty$.  Furthermore, because $w$ has period 1,
\begin{equation*}
  2\int_{z_{i}}^{z_{i+1}}\sqrt{w\left(s\right)}\,ds
  =
  \sum_{j=0}^{z_{i+1}-z_{i}-1}
  2\int_{z_{i}+j}^{z_{i}+j+1}\sqrt{w\left(s\right)}\,ds
  =
  (z_{i+1}-z_{i})\bar{p},  
\end{equation*}
where we have assumed that $z_{i+1}>z_{i}$.  Therefore
\begin{equation*}
  \underset{n\rightarrow \infty}{\liminf}
  \sum_{i=1}^{n}
  \left[
  \varepsilon_{n}\int_{t_{i}-\gamma}^{t_{i}+\gamma}\xi_{n}^\prime(x)^{2}\,dx
  +
  \varepsilon_{n}^{-1}\int_{t_{i}-\gamma}^{t_{i}+\gamma}w\left(\xi_{n}(x)\right)\,dx
  \right]
  \geq
  \operatorname{Var}(\xi^{*},\reals)\bar{p}.
\end{equation*}
%



{\bf Step 2.}
We construct a recovery sequence for a typical element in
$\bar{A}$.
This construction proceeds in several steps.  First, we build a recovery
sequence for the unit step function.  With this basic case,
a recovery sequence is constructed for a piecewise constant
function having a single jump of size $K$ at a point $x_{0}$.  Lastly,
using the previous construction, we show how to connect recovery
sequences for isolated jumps, which allows us to build a recovery
sequence for any function $\xi^{*}\in BV(\reals;\mathbb{Z})$.

{\bf Step 2a.}
Let $H$ denote the unit step function.
Our initial goal is to construct a recovery sequence for $H$.
This sequence is more complicated than necessary for its immediate
purpose, because
later we use this construction as a building block for a 
recovery sequence for steps larger than 1 unit.
Let $\bar{\xi}$ be the solution to 
$\xi'=\sqrt{w\left(\xi\right)}$ with $\xi(0)=1/2$.
A straightforward argument shows that $\bar{\xi}-1/2$ is an odd
function because $w$ is even about 
$1/2$.  
Also, $\bar{\xi}$ is strictly increasing and 
$\bar{\xi}(x)\rightarrow  1$ as $x\rightarrow \infty$.

Define $\bar{\xi}_{n}(x)=\bar{\xi}(x/\varepsilon_{n})$.
Let $L_{n}^{+}$
denote the function whose graph is the line tangent to 
$\bar{\xi}_{n}$ at
$\left(\varepsilon_{n}^{1/2},\bar{\xi}_{n}(\varepsilon_{n}^{1/2})\right)$ and
let
$(\beta_{n},1)$ be the point where $L_{n}^{+}$
crosses the line $y=1$.  
Likewise, let
$L_{n}^{-}(x)$ be the function whose graph is the line tangent to 
$\bar{\xi}_{n}$ at
$\left(-\varepsilon_{n}^{1/2},\bar{\xi}_{n}(-\varepsilon_{n}^{1/2})\right)$.
By symmetry,
$L_{n}^{-}$ crosses the $y$-axis at
$(-\beta_{n},0)$.
Next we define $\hat{\xi}_{n}$ by
\begin{equation*}
  \hat{\xi}_{n}(x)
  :=
  \left\{
  \begin{array}{ll}
    0, & x\leq -\beta_{n}, \\[1mm]
    L_{n}^{-}(x), &
     -\beta_{n} < x \leq -\varepsilon_{n}^{1/2},\\[1mm]
    \bar{\xi}_{n}(x), &
    -\varepsilon_{n}^{1/2} < x \leq \varepsilon_{n}^{1/2},\\[1mm]
    L_{n}^{+}(x), &
    \varepsilon_{n}^{1/2} < x \leq \beta_{n},\\[1mm]
    1, &
    x > \beta_{n}.\\        
  \end{array}
  \right. 
\end{equation*}
See Figure~\ref{ff2}.
Now we let $\xi_{n}$ be the function that is piecewise affine on
$P_{\varepsilon_{n}\delta_{n}}$ such that
$\xi_{n}(i\varepsilon_{n}\delta_{n})=\hat{\xi}_{n}(i\varepsilon_{n}\delta_{n})$
for all $i\in \mathbb{Z}$.  We show that $\{ \xi_{n} \}$ is a recovery
sequence for $H$.

\begin{figure}[htb]
\centering
  \includegraphics[scale = 0.4, clip, trim=0in .35in 0in .25in]{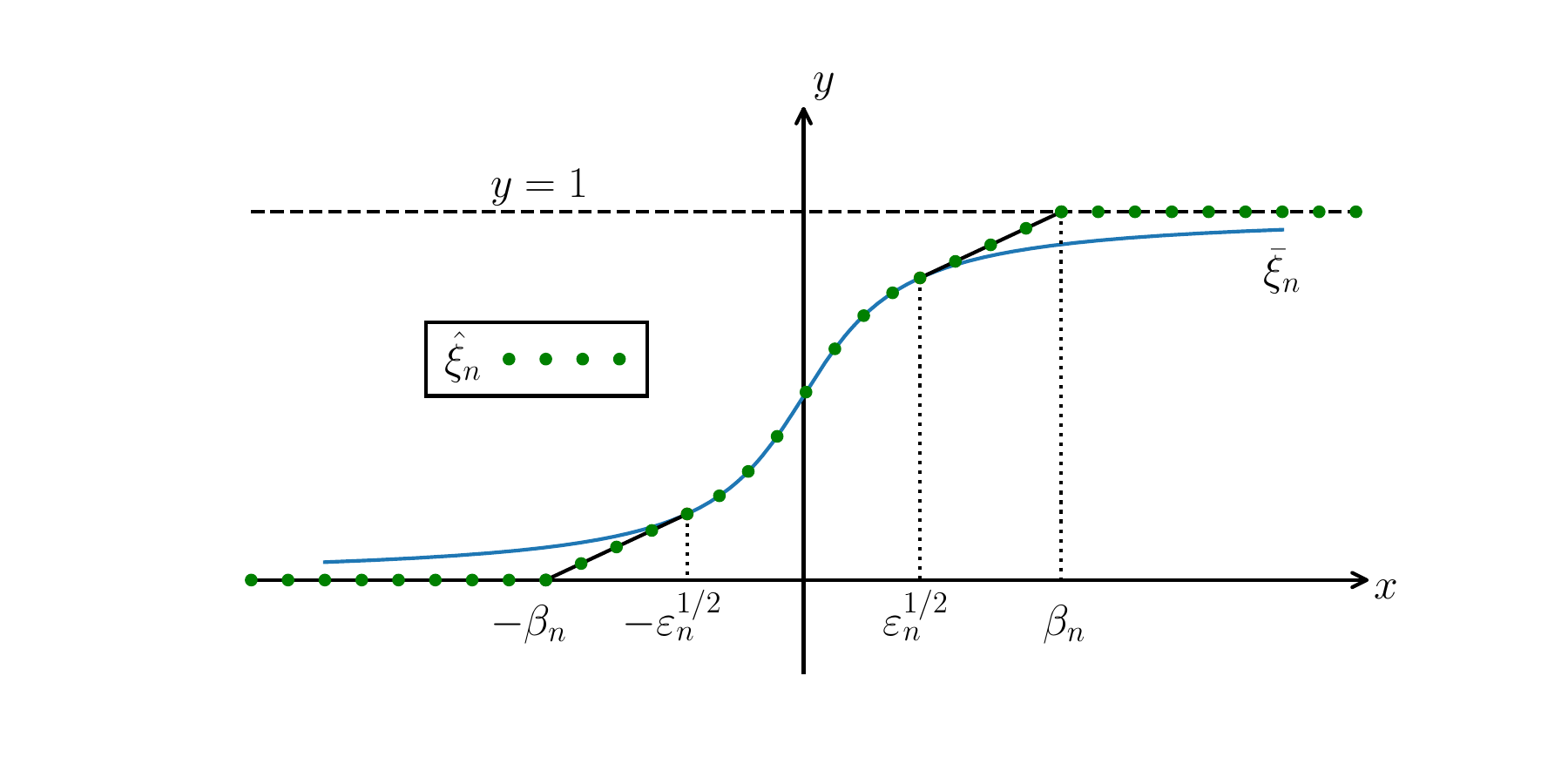}
  \caption{Construction of $\hat{\xi}_{n}$.}
  \label{ff2}
\end{figure}

We start by showing that
%
%
$\{ \xi_{n} \}$ converges to $H$ in $L^{1}_{\text{loc}}(\reals)$.
It is sufficient to show that $\beta_{n}\rightarrow 0$
as $n\rightarrow \infty$.
Because
$\left(\beta_{n},1\right)$ is on the tangent line,
$\beta_{n}
=
(1-\bar{\xi}_{n}(\varepsilon_{n}^{1/2}))/\bar{\xi}'_{n}(\varepsilon_{n}^{1/2})
+
\varepsilon_{n}^{1/2}$.  Next,
\begin{equation}
  \frac{1-\bar{\xi}_{n}(\varepsilon_{n}^{1/2})}{\bar{\xi}'_{n}(\varepsilon_{n}^{1/2})}
  =
  \frac{\varepsilon_{n}(1-\bar{\xi}(\varepsilon_{n}^{-1/2}))}
       {\bar{\xi}'(\varepsilon_{n}^{-1/2})} 
  =
  \frac{\varepsilon_{n}(1-\bar{\xi}(\varepsilon_{n}^{-1/2}))}
       { \sqrt{ w\big(\bar{\xi}(\varepsilon_{n}^{-1/2})\big)}}.
  \label{ee174}
\end{equation}
%
To estimate this, we look at
\begin{equation}
  w\big(\bar{\xi}(\varepsilon_{n}^{-1/2})\big)
  =
  \int_{\bar{\xi}(\varepsilon_{n}^{-1/2})}^{1}\!
  \int_{s}^{1}\! w''(t)\,dt\, ds
  \geq
  \frac{w^{\#}}{2}
  \left(
    1-\bar{\xi}(\varepsilon_{n}^{-1/2})
  \right)^{2},
  \label{ee175}
\end{equation}
where $w^{\#}$ is chosen so that $w''(\xi)\geq w^{\#}>0$ for $\xi$ in an interval
to the left of $1$.  Using
\eqref{ee175} in \eqref{ee174} gives
\begin{equation}
  \frac{1-\bar{\xi}_{n}(\varepsilon_{n}^{1/2})}{\bar{\xi}'_{n}(\varepsilon_{n}^{1/2})}
  \leq
  \frac{\sqrt{2}\varepsilon_{n}}{\sqrt{w^{\#}}}
  \rightarrow 0.
  \label{ee176}
\end{equation}

Now we consider the limiting behavior of
$E_{\varepsilon_{n},\delta_{n}}[\xi_{n}]$.
We note first that, because
$\bar{\xi}'=\sqrt{w\left(\bar{\xi}\right)}$, 
\begin{equation}
  \int_{-\infty}^{\infty}\!
  \left[
    \bar{\xi}'(y)^{2}
    +
    w\left(\bar{\xi}(y)\right)
  \right]
  \,dy 
  =
  \int_{-\infty}^{\infty}\!
  2\sqrt{w\left(\bar{\xi}(y)\right)}\bar{\xi}'(y)\,dy
  =
  \int_{0}^{1}\!
  2\sqrt{w(s)}\,ds
  =
  \bar{p}.
  \label{ee177}
\end{equation}
Next we show that
\begin{equation}
  \limsup \varepsilon_{n}
  \int_{-\infty}^{\infty}\!\xi_{n}'(x)^{2} \,dx
  \leq
  \int_{-\infty}^{\infty}\!\bar{\xi}'(y)^{2} \,dy.
  \label{ee182}
\end{equation}
The same argument as in \eqref{ee48} yields
\begin{equation*}
  \int_{i\varepsilon_{n}\delta_{n}}^{(i+1)\varepsilon_{n}\delta_{n}}\!
  \xi_{n}'(x)^{2}\,dx
  \leq
  \int_{i\varepsilon_{n}\delta_{n}}^{(i+1)\varepsilon_{n}\delta_{n}}\!
  \hat{\xi}_{n}'(x)^{2}\,dx,
\end{equation*}
whence
\begin{equation}
  \limsup \varepsilon_{n}
  \int_{-\infty}^{\infty}\!\xi_{n}'(x)^{2} \,dx
  \leq
  \limsup \varepsilon_{n}
  \int_{-\infty}^{\infty}\!
  \hat{\xi}_{n}'(x)^{2}\,dx.
  \label{ee184}
\end{equation}
The integral on the right-hand side of \eqref{ee184} equals
\begin{equation}
  \varepsilon_{n}\int_{-\beta_{n}}^{-\varepsilon_{n}^{1/2}}\!
  [\left(L_{n}^{-}\right)'\!(x)]^{2}\,dx
  +
  \varepsilon_{n}\int_{\varepsilon_{n}^{1/2}}^{\beta_{n}}\!
  [\left(L_{n}^{+}\right)'\!(x)]^{2}\,dx
  +
  \varepsilon_{n}\int_{-\varepsilon_{n}^{1/2}}^{\varepsilon_{n}^{1/2}}\!
  \bar{\xi}_{n}'(x)^{2}\,dx.
  \label{ee185}
\end{equation}
After making the change of variables
$y=x/\varepsilon_{n}$, it is easy to see that the last integral in
\eqref{ee185} converges to
$
\int_{-\infty}^{\infty}\!\bar{\xi}'(y)^{2} \,dy
$.
The second integral in \eqref{ee185} equals
\begin{equation*}
\begin{split}
  \varepsilon_{n}
  \left(\beta_{n}-\varepsilon_{n}^{^{1/2}}\right)
  \left(
   \bar{\xi}_{n}'(\varepsilon_{n}^{1/2})
  \right)^{2}
  &=
  \varepsilon_{n}
  \left(
   1-\bar{\xi}_{n}(\varepsilon_{n}^{1/2})
  \right)
  \bar{\xi}_{n}'(\varepsilon_{n}^{1/2})   \\
  &=
  \left(
   1-\bar{\xi}(\varepsilon_{n}^{-1/2})
  \right)  
  \bar{\xi}'(\varepsilon_{n}^{-1/2})
  \rightarrow 0.
\end{split}
%
\end{equation*}
The estimate for the first integral in \eqref{ee185} is similar.  
The inequality in \eqref{ee182} follows.

The limiting behavior of the interaction term in $E_{\varepsilon_{n},\delta_{n}}[\xi_{n}]$ is
considered next.  For this, we define $K_{n}\in \mathbb{Z}$ such that
$K_{n}\varepsilon_{n}\delta_{n}\leq \varepsilon_{n}^{1/2}<
(K_{n}+1)\varepsilon_{n}\delta_{n}$ and $J_{n}\in
\mathbb{Z}$ such that $J_{n}\varepsilon_{n}\delta_{n}\leq \beta_{n}<
(J_{n}+1)\varepsilon_{n}\delta_{n}$.
We show that 
\begin{equation}
  \varepsilon_{n}^{-1}
  \int_{-K_{n}\varepsilon_{n}\delta_{n}}^{K_{n}\varepsilon_{n}\delta_{n}}\!
    w\left(\xi_{n}(x)\right)\,dx
  \rightarrow 
  \int_{-\infty}^{\infty}\!
  w\left(\bar{\xi}(y)\right)\,dy
  \label{ee178}
\end{equation}
as $n\rightarrow \infty$.  Making the change of variables
$y=x/\varepsilon_{n}$ in the integral on the left in \eqref{ee178}
gives
$$
\int_{-K_{n}\delta_{n}}^{K_{n}\delta_{n}}\!
    w\left(\xi_{n}(\varepsilon_{n}y)\right)\,dy.$$
The function $y\mapsto\xi_{n}(\varepsilon_{n}y)$ is continuous and piecewise
affine on $P_{\delta_{n}}$, and
$\xi_{n}(\varepsilon_{n}i\delta_{n})=\bar{\xi}(i\delta_{n})$
for $|i|\leq K_{n}$.  
Hence for any $M$ such that $K_{n}\delta_{n}\geq M>0$,
$w\left(\xi_{n}(\varepsilon_{n}y)\right)$ converges uniformly
to $w\left(\bar{\xi}(y)\right)$ on $[-M,M]$.

Let $\mu>0$ and choose $M$ large enough so that 
$
\int_{M}^{\infty}\!
    w\left(\bar{\xi}(y)\right)dy<\mu
$.
(That $w\left(\bar{\xi}(y)\right)$ is integrable follows from \eqref{ee177}).
We know that
$\varepsilon_{n}^{1/2}<(K_{n}+1)\varepsilon_{n}\delta_{n}$, which
implies that 
$K_{n}\delta_{n}>\varepsilon_{n}^{-1/2}-\delta_{n}$.  Choose $N$ large
enough so that $n>N$ implies that $K_{n}\delta_{n}>M$.    
Then
\begin{equation}
\begin{split}
  \left|\int_{-\infty}^{\infty}\!\right.
  &
  w\left(\bar{\xi}(y)\right)\,dy
  -
  \left.
  \int_{-K_{n}\delta_{n}}^{K_{n}\delta_{n}}\!
  w\left(\xi_{n}(\varepsilon_{n}y)\right)\,dy
  \right| \\
  &\leq
  \int_{-M}^{M}\!
  \left|w\left(\bar{\xi}(y)\right)
  -
  w\left(\xi_{n}(\varepsilon_{n}y)\right)
  \right|\,dy
  +
  2\mu
  +
  2\left|
  \int_{M}^{K_{n}\delta_{n}}\!
  w\left(\xi_{n}(\varepsilon_{n}y)\right)\,dy
  \right|.  
\end{split}
  \label{ee179}  
\end{equation}
We need to estimate the last term on the right-hand side of
\eqref{ee179}.  
We can suppose that $M=\hat{K}_{n}\delta_{n}$ for some
$\hat{K}_{n}\in\mathbb{Z}$.  If $i\geq\hat{K}_{n}$ and 
$y\in [i\delta_{n},(i+1)\delta_{n}]$, then
$1/2
  \leq
  \xi_{n}(i\varepsilon_{n}\delta_{n})
  \leq
  \xi_{n}(\varepsilon_{n}y)
  \leq
1$,
which implies that
$w\left(\xi_{n}(\varepsilon_{n}y)\right)
  \leq
  w\left(\xi_{n}(i\varepsilon_{n}\delta_{n})\right)
  =
  w\left(\bar{\xi}(i\delta_{n})\right)
$.
Therefore
\begin{align}
  \int_{M}^{K_{n}\delta_{n}}\!
  w\left(\xi_{n}(\varepsilon_{n}y)\right)\,dy
  &=
  \sum_{i=\hat{K}_{n}}^{K_{n}-1}
  \int_{i\delta_{n}}^{(i+1)\delta_{n}}\!
  w\left(\xi_{n}(\varepsilon_{n}y)\right)\,dy
  \nonumber \\
  &\leq
  \sum_{i=\hat{K}_{n}}^{K_{n}-1}
  \int_{i\delta_{n}}^{(i+1)\delta_{n}}\!
  w\left(\bar{\xi}(i\delta_{n})\right)\,dy  \label{ee181}\\
  &\leq
  \sum_{i=\hat{K}_{n}}^{\infty}
  w\left(\bar{\xi}(i\delta_{n})\right)\delta_{n}
  \leq
  \int_{(\hat{K}_{n}-1)\delta_{n}}^{\infty}\!
  w\left(\bar{\xi}(y)\right)\,dy
  <
  \delta_{n}+\mu
  \nonumber
\end{align}
Returning to \eqref{ee179} and using \eqref{ee181}, we see that  
$\underset{n\rightarrow \infty}{\limsup}$ of the left-hand side of 
\eqref{ee179} is less than or equal to $4\mu$.  Since $\mu>0$ is
arbitrary, we conclude that \eqref{ee178} holds.

One more estimate establishes the limiting behavior of the interaction
term in $E_{\varepsilon_{n},\delta_{n}}[\xi_{n}]$.
Using that
$(J_{n}+1)\delta_{n}<\varepsilon_{n}^{-1}\beta_{n}+\delta_{n}$
and
$\varepsilon_{n}^{-1/2}<K_{n}\delta_{n}+\delta_{n}$,
we have 
\begin{align}
  \varepsilon_{n}^{-1}
  \int_{K_{n}\varepsilon_{n}\delta_{n}}^{\infty}\!
    w\left(\xi_{n}(x)\right)\,dx
  &=
  \varepsilon_{n}^{-1}
  \int_{K_{n}\varepsilon_{n}\delta_{n}}^{(J_{n}+1)\varepsilon_{n}\delta_{n}}\!
    w\left(\xi_{n}(x)\right)\,dx \nonumber \\
  &\leq
    \left[
      (J_{n}+1)\delta_{n}-K_{n}\delta_{n}
    \right]
    w\left(\xi_{n}(K_{n}\varepsilon_{n}\delta_{n}) \right) \nonumber \\[1mm]
  &\leq
    (\varepsilon_{n}^{-1}\beta_{n}-\varepsilon_{n}^{-1/2})
    w\left(\bar{\xi}(K_{n}\delta_{n})\right)
    +
    2\delta_{n}w\left(\bar{\xi}(K_{n}\delta_{n})\right)\label{ee187}\\[1mm]
  &=
  \frac{(1-\bar{\xi}(\varepsilon_{n}^{-1/2}))}
       { \sqrt{ w\big(\bar{\xi}(\varepsilon_{n}^{-1/2})\big)}}
       w\left(\bar{\xi}(K_{n}\delta_{n})\right)
       +
       2\delta_{n}w\left(\bar{\xi}(K_{n}\delta_{n})\right),
  \nonumber 
\end{align}
where the final equality uses \eqref{ee174}.  For the last line 
of \eqref{ee187}, we know that
$w\left(\bar{\xi}(K_{n}\delta_{n})\right)\rightarrow 0$ as
$n\rightarrow \infty$ and we know by \eqref{ee176}
that the remaining part of the first term is
bounded.  
A similar estimate
applies to
$\varepsilon_{n}^{-1}
  \int_{-\infty}^{-K_{n}\varepsilon_{n}\delta_{n}}\!
    w\left(\xi_{n}\right)\,dx$.
This completes the construction of a recovery sequence for $H$, the
unit step function.

{\bf Step 2b.}
We now use the sequence just constructed to build a 
recovery sequence 
for a jump of size
$K\in\mathbb{Z}$ at an arbitrary point $x_{0}\in\reals$.  To
illustrate, we construct a recovery sequence for $2H$, the
jump of size 2 at the origin.  
Let $\hat{\xi}_{n}$ and $\xi_{n}$ be defined as above.  Let
$\{ t_{n} \}$ be a sequence converging to 0 such that
$t_{n}=T_{n}\varepsilon_{n}\delta_{n}$ with $T_{n}\in\mathbb{Z}$ and
$t_{n}>2\beta_{n}+\varepsilon_{n}\delta_{n}$ for all $n$.  Set
$\lambda_{n}(x)=\xi_{n}(x)+\xi_{n}(x-t_{n})$.  We claim that
$\{ \lambda_{n}\}$ is a recovery sequence for $2H$.  
First we note that $0\leq \lambda_{n}(x)\leq 2$ for all $x$, that $\lambda_{n}(x)=0$
for $x\leq(-J_{n}-1)\varepsilon_{n}\delta_{n}\rightarrow 0$, and that $\lambda_{n}(x)=2$
for $x\geq(T_{n}+J_{n}+1)\varepsilon_{n}\delta_{n}\rightarrow 0$.  Hence
$\{ \lambda_{n} \}$ converges to $2H$ in $L^{1}_{\text{loc}}(\reals)$.


For the energy, we have
\begin{align*}
  E_{\varepsilon_{n},\delta_{n}}[\lambda_{n}]
  &=
  \int_{(-J_{n}-1)\varepsilon_{n}\delta_{n}}^{(J_{n}+1)\varepsilon_{n}\delta_{n}}\!
  \left[
    \varepsilon_{n}\lambda'_{n}(x)^{2}
    +
    \varepsilon_{n}^{-1}w(\lambda_{n}(x))
  \right]
  \,dx \nonumber \\
  &\phantom{mmmmmmm}
  +
  \int_{(T_{n}-J_{n}-1)\varepsilon_{n}\delta_{n}}^{(T_{n}+J_{n}+1)\varepsilon_{n}\delta_{n}}\!
  \left[
    \varepsilon_{n}\lambda'_{n}(x)^{2}
    +
    \varepsilon_{n}^{-1}w(\lambda_{n}(x))
  \right]
  \,dx  \nonumber \\[2mm]
  &=
  E_{\varepsilon_{n},\delta_{n}}[\xi_{n}]
  +
  \int_{(T_{n}-J_{n}-1)\varepsilon_{n}\delta_{n}}^{(T_{n}+J_{n}+1)\varepsilon_{n}\delta_{n}}\!
  \Big[
    \varepsilon_{n}
    \left[
    \left(
      \xi_{n}(x-t_{n})+1
    \right)'
    \right]^{2}
    \nonumber \\  
    &\phantom{mmmmmmmmmmmmmmn}
    +\,
    \varepsilon_{n}^{-1}
    w\left(\xi_{n}(x-t_{n})+1\right)
  \Big]
  \,dx  
  \\
  &=E_{\varepsilon_{n},\delta_{n}}[\xi_{n}]
  +
  \int_{(-J_{n}-1)\varepsilon_{n}\delta_{n}}^{(J_{n}+1)\varepsilon_{n}\delta_{n}}\!
  \left[
    \varepsilon_{n}\xi'_{n}(y)^{2}
    +
    \varepsilon_{n}^{-1}w(\xi_{n}(y))
  \right]
  \,dy \nonumber \\
  &=
  2E_{\varepsilon_{n},\delta_{n}}[\xi_{n}] \nonumber,  
\end{align*}
where for the third equality we use that $w$ has period 1 and we make
the change of variables $y=x-t_{n}$.  It follows that
$\limsup E_{\varepsilon_{n},\delta_{n}}[\lambda_{n}]\geq 2\bar{p}$.

Let $K$ be a positive integer.
Generalizing the previous construction,
we can build a recovery sequence for
a jump up from $0$ to
$K$ at $x=0$
by horizontally translating $K$ recovery sequences
each for a 
jump of size 1.  
Furthermore, the functions in that recovery sequence can be vertically
translated by $K'\in \mathbb{Z}$ to build a recovery sequence for a
jump up from $K'$ to $K'+K$ at $x=0$.
For a jump at a point $x_{0}$ different from 0,
we pick
$x_{n}\in P_{\varepsilon_{n}\delta_{n}}$ such that
$x_{n}\rightarrow x_{0}$.  If $\{ \xi_{n} \}$ is a recovery sequence for a
jump from $K'$ to $K'+K$ at $x=0$, then $\{ \xi_{n}(x-x_{n})  \}$ is a
recovery sequence for the corresponding jump at $x=x_{0}$.
(Assuming that $x_{n}\in P_{\varepsilon_{n}\delta_{n}}$ ensures that
$\xi_{n}(x-x_{n})\in A_{\varepsilon_{n}\delta_{n}}$.)

Let $K$ be a negative integer.
For a function with a jump down from $K'$ to $K'+K$ at $x=x_{0}$,
we build a recovery sequence by
repeating the above constructions but
starting with the solution to
$\xi'=-\sqrt{w\left(\xi\right)}$, $\xi(0)=1/2$.
It is easy to check that the key equality \eqref{ee177} holds
if $\xi'=-\sqrt{w\left(\xi\right)}$.

{\bf Step 2c.}
Lastly, we create a recovery sequence for an arbitrary function 
$\xi\in BV(\reals;\mathbb{Z})$.
There are numbers $t_{1}<t_{2}<\cdots<t_{m}$ and
integers $K_{1},\ldots,K_{m+1}$ such that
$\xi(x)=K_{1}+\sum_{j=1}^{m}(K_{j+1}-K_{j})H(x-t_{j})$.  
To define the typical function in our recovery sequence, near each
jump we use a function from the recovery sequence for that jump, as
constructed above.  On
the interval between 2 jumps, we set our typical function equal 
to the appropriate constant value.
Specifically, let $\{ \xi_{n}^{j} \}_{n}$ be the recovery sequence we constructed
above for
$K_{j}+(K_{j+1}-K_{j})H(x-t_{j})$, which is the jump from $K_{j}$ to
$K_{j+1}$ at $x=t_{j}$.
We know that
there are sequences $\{ \alpha_{n}^{j} \}_{n}$ and $\{ \beta_{n}^{j} \}_{n}$
with $\alpha_{n}^{j}, \beta_{n}^{j}\in P_{\varepsilon_{n}\delta_{n}}$
such that
\begin{equation*}
  \xi_{n}^{j}(x)
  =
  \begin{cases}
  K_{j} &   \text{for}\ x\leq \alpha_{n}^{j}  \\
  K_{j+1} & \text{for}\  x\geq \beta_{n}^{j}
  \end{cases}
  \quad \text{and} \quad
  \alpha_{n}^{j}\nearrow t_{j},\ \ 
  \beta_{n}^{j}\searrow t_{j}.  
\end{equation*}
We pick $N$ large enough such that $n>N$ implies that
$\alpha_{n}^{j}>\beta_{n}^{j-1}$ for $j=2,\ldots,m$.  Then for
$n>N$, we define
\begin{equation*}
  \xi_{n}(x)
  =
  \begin{cases}
  \xi_{n}^{j}(x) & \text{for}\ \alpha_{n}^{j}\leq x\leq \beta_{n}^{j},  \\
  \xi(x) & \text{otherwise}.
  \end{cases}
\end{equation*}
Because $\xi(x)=K_{j}\in \mathbb{Z}$ for
$\beta_{n}^{j-1}\leq x \leq \alpha_{n}^{j}$,
$E_{\varepsilon_{n}\delta_{n}}[\xi_{n}]=\sum_{j=1}^{m}E_{\varepsilon_{n}\delta_{n}}[\xi_{n}^{j}]$.
Hence
\begin{equation*}
  \limsup E_{\varepsilon_{n}\delta_{n}}[\xi_{n}]
  \leq
  \sum_{j=1}^{m}
  \limsup E_{\varepsilon_{n}\delta_{n}}[\xi_{n}^{j}]
  \leq
  \sum_{j=1}^{m}|K_{j+1}-K_{j}|\bar{p}
  =
  \operatorname{Var}(\xi,\reals)\bar{p}.
\end{equation*}
%
\hfill $\Box$

\section{Conclusion} 



Starting from a Frenkel-Kontorova-type model of an
infinitely long one-dimensional chain of atoms weakly interacting with
a substrate of fixed atoms,
we derive a rescaled model containing a small parameter
$\delta$ that measures the relative strengths of the weak interaction
and the elastic interaction.
We then apply a discrete-to-continuum approach, replacing discrete
displacements with piecewise affine functions to define continuum
versions of the discrete energies.
In Theorem~\ref{r1}, we
prove that these continuum energies $\Gamma$-converge as $\delta\rightarrow 0$.
We interpret this limiting process as a transition from the microscale
to a mesoscale at which a single
diffuse domain wall is observed.
We next introduce an additional rescaling $\varepsilon$,
and an associated limiting process, that converts our problem to the
macroscale.  In this case the limiting energy is finite for piecewise constant
functions of bounded variation.
For this limiting energy, each
point of discontinuity of a minimizer
corresponds to a sharp domain wall.  

We relate
Theorem~\ref{r1} in this paper to the modeling and numerical results in our
earlier paper \cite{espanol2017discrete}.  In that paper, we developed a discrete
model similar to but more general than the discrete model presented in
Section~\ref{s1}.  The model in \cite{espanol2017discrete} allows atoms on the
chain to deflect vertically as well as horizontally, so that the chain
can bend.  The corresponding energy includes an additional elastic
term that penalizes bending.  In \cite{espanol2017discrete}, we derive from the
discrete model a formal continuum limit based on a small geometric
parameter that measures the ratio of the atomic spacing to the lateral
extent of the system.  We then perform numerical simulations to
compare the predictions of the discrete and continuum models.  These
numerical results suggest that for the limiting process based on this
geometric parameter, atoms are homogenized and the chain is accurately
decribed by a continuous curve.

Furthermore, in \cite{espanol2017discrete} we use numerical simulations to
explore how varying the relative strengths of the elastic terms
compared to the strength of the interaction term in the energy affects
the structure of domain walls.  We observe that if elastic interactions are
relatively strong compared to the weak interactions, then the domain
walls are spatially diffuse rather than concentrated and, as a
consequence, the domain walls are composed of relatively many atoms.
If, on the other hand, the elastic interactions are relatively weak compared
to the weak interactions, then domain walls are sharp and are composed
of relatively few atoms.
In this paper, letting $\delta\rightarrow 0$ corresponds to the
limiting case of the ratio of the strength of the elastic term to the
strength of the interaction term going to infinity.
Hence we can view
Theorem~\ref{r1} above as 
providing a rigorous justification for the
continuum limit that was postulated in \cite{espanol2017discrete}, at least for
the special case of the model considered here.

%

%
%
%
%


\bibliography{gamma-converg-references}	

\end{document}